\documentclass[a4paper,english]{lipics-v2018}

\usepackage{microtype}\usepackage{ifthen}
\usepackage{amsmath,array,bm}
\usepackage{amsmath,amsfonts}
\usepackage{amsthm}
\usepackage{amssymb,latexsym}
\usepackage{algorithm}
\usepackage{subcaption}
\usepackage{algorithmicx}
\usepackage[]{algpseudocode}
\usepackage[algo2e,linesnumbered]{algorithm2e}
\usepackage[dvipsnames]{xcolor}
\usepackage{graphicx}
\usepackage{caption}
\usepackage{enumitem}
\usepackage{float}
\usepackage{hyperref}
\usepackage[square,sort,comma,numbers]{natbib}
\usepackage{tikz}
\usepackage{xspace}
\usepackage{xcolor}
\usepackage{mdframed}
\usepackage{mathtools}
\usepackage[export]{adjustbox}[2011/08/13]
\usepackage{multirow}
\usepackage{hhline}

\theoremstyle{theorem}

\newtheorem{claim}[theorem]{Claim}

\newboolean{short}
\setboolean{short}{false}

\title{Large-Scale Distributed Algorithms for Facility Location with Outliers}
\author{Tanmay Inamdar}{Department of Computer Science, The University of Iowa, Iowa, USA}{tanmay-inamdar@uiowa.edu}{}{}
\author{Shreyas Pai}{Department of Computer Science, The University of Iowa, Iowa, USA}{shreyas-pai@uiowa.edu}{}{}
\author{Sriram V.~Pemmaraju}{Department of Computer Science, The University of Iowa, Iowa, USA}{sriram-pemmaraju@uiowa.edu}{}{}

\authorrunning{T.~Inamdar, S.~Pai, S.~V.~Pemmaraju}

\Copyright{Tanmay Inamdar, Shreyas Pai, and Sriram V.~Pemmaraju}

\subjclass{\ccsdesc[500]{Theory of computation~Facility location and clustering},
\ccsdesc[500]{Theory of computation~Distributed algorithms},
\ccsdesc[500]{Theory of computation~MapReduce algorithms},
\ccsdesc[300]{Theory of computation~Graph algorithms analysis}}

\keywords{Distributed Algorithms,
Clustering with Outliers,
Metric Facility Location,
Massively Parallel Computation,
k-machine model,
Congested Clique
}

\category{}

\relatedversion{}

\supplement{}

\funding{}

\acknowledgements{}

\EventEditors{Jiannong Cao, Faith Ellen, Luis Rodrigues, and Bernardo Ferreira}
\EventNoEds{4}
\EventLongTitle{22nd International Conference on Principles of Distributed Systems (OPODIS 2018)}
\EventShortTitle{}
\EventAcronym{OPODIS}
\EventYear{2018}
\EventDate{December 17–19, 2018}
\EventLocation{Hong Kong, China}
\EventLogo{}
\SeriesVolume{122}
\ArticleNo{} \nolinenumbers \hideLIPIcs

\newcommand{\q}{\mathbf{q}}
\renewcommand{\i}{\iota}

\renewcommand{\ij}{\i j}

\renewcommand{\cdots}{\ldots}
\newcommand{\lr}[1]{\left(#1\right)}
\newcommand{\LR}[1]{\left\{#1\right\}}

\newcommand{\lra}{Lenzen's routing protocol}
\renewcommand{\ll}{\log\log n}
\renewcommand{\lll}{\log\log\log n}

\DeclareMathOperator{\poly}{poly}

\DeclareMathAlphabet{\mathpzc}{OT1}{pzc}{m}{it}

\newcommand{\cost}{\mathsf{cost}}
\renewcommand{\epsilon}{\varepsilon}
\newcommand{\LP}{\textsf{LP}\xspace}

\newcommand{\kmm}{\(k\)-machine model\xspace}
\renewcommand{\r}{\mbox{rank}}
\date{}
\begin{document}
\maketitle

\begin{abstract}
This paper presents fast, distributed, $O(1)$-approximation algorithms for
metric facility location problems with outliers in the Congested Clique model,
Massively Parallel Computation (MPC) model, and in the $k$-machine model.
The paper considers \textit{Robust Facility Location} and \textit{Facility
Location with Penalties}, two versions of the facility location problem
with outliers proposed by Charikar et al. (SODA 2001). The paper also considers
two alternatives for specifying the input: the input metric can be provided
explicitly (as an $n \times n$ matrix distributed among the machines) or
implicitly as the shortest path metric of a given edge-weighted graph. The
results in the paper are:

\begin{itemize}
\item \textbf{Implicit metric}: For both problems, $O(1)$-approximation
algorithms running in $O(\mbox{poly}(\log n))$ rounds in the Congested Clique
and the MPC model and $O(1)$-approximation algorithms running in
$\tilde{O}(n/k)$ rounds in the $k$-machine model.

\item \textbf{Explicit metric}: For both problems, $O(1)$-approximation
algorithms running in $O(\log\log\log n)$ rounds in the Congested Clique and
the MPC model and $O(1)$-approximation algorithms running in
$\tilde{O}(n/k)$ rounds in the $k$-machine model.
\end{itemize}

Our main contribution is to show the existence of Mettu-Plaxton-style
$O(1)$-approximation algorithms for both Facility Location with outlier
problems. As shown in our previous work (Berns et al., ICALP 2012,
Bandyapadhyay et al., ICDCN 2018) Mettu-Plaxton style algorithms are more
easily amenable to being implemented efficiently in distributed and
large-scale models of computation.

\end{abstract}

\ifthenelse{\boolean{short}}{}{
\newpage
\tableofcontents
\newpage
}
\section{Introduction}

\textit{Metric Facility Location} (in short, \textsc{FacLoc}) is a well-known combinatorial optimization problem
used to model clustering problems. The input to the problem is a set $F$ of \textit{facilities},
an \textit{opening cost} $f_i \ge 0$ for each facility $i \in F$, 
a set $C$ of \textit{clients}, and a metric space $(F \cup C, d)$ of \textit{connection costs}, where $d(i, j)$ denotes the
cost of client $j$ connecting to facility $i$.
The objective is to find a subset $F' \subseteq F$ of facilities to open so that the total cost of 
opening the facilities plus the cost of connecting all clients to open facilities is
minimized. In other words, the quantity $cost(F') := \sum_{i \in F'} f_i + \sum_{j \in C} d(j, F')$ is minimized, 
where $d(j, F')$ denotes $\min_{i \in F'} d(i, j)$. \textsc{FacLoc} is NP-complete, but
researchers have devised a number of approximation algorithms for the problem. For any
$\alpha \ge 1$, an \textit{$\alpha$-approximation algorithm} for \textsc{FacLoc} finds in polynomial
time, a subset $F' \subseteq F$ of facilities such that $cost(F') \le \alpha \cdot cost(F^*)$,
where $F^*$ is an optimal solution to the given instance of \textsc{FacLoc}.
There are several well-known $O(1)$-factor approximation algorithms for \textsc{FacLoc} including the primal-dual algorithm of Jain and Vazirani \cite{JainVaziraniJACM2001} and the greedy
algorithm of Mettu and Plaxton \cite{MettuPlaxtonSICOMP2003}. The best approximation factor
currently achieved by an algorithm for FacLoc is 1.488 \cite{LiICALP2011}.
More recently, motivated by the need to solve \textsc{FacLoc} and other clustering problems on extremely large inputs,
researchers have proposed distributed and parallel approximation algorithms for these problems.
See for example \cite{EneIMKDD2011,GarimellaDGSCIKM2015} for clustering algorithms in systems such as MapReduce \cite{DeanGhemawatCACM2010}
and Pregel \cite{MalewiczABDHLCSIGMOD2010} and \cite{BandyapadhyayArxiv2017} for clustering algorithms in the \textit{$k$-machine model}. 
Clustering algorithms \cite{SilvaFBHCCompSurv2013} have also been designed for streaming models of computation 
\cite{AlonMSSTOC1996}. 

Outliers can pose a problem for many statistical methods. For clustering problems, a few outliers can have an outsized influence
on the optimal solution, forcing the opening of costly extra facilities or leading to poorer service to many clients.
Versions of \textsc{FacLoc} that are robust to outliers have been proposed by Charikar et al.~\cite{CharikarSODA2001}, where the authors also
present $O(1)$-approximation algorithms for these problems. Specifically, Charikar et al.~\cite{CharikarSODA2001} propose
two versions of \textsc{FacLoc} that are robust to outliers:
\begin{description}
\item[Robust \textsc{FacLoc}:] In addition to $F$, $C$, opening costs $\{f_i | i \in F\}$, and metric $d$, 
we are also given an integer $0 \le p \le |C|$, that denotes the \textit{coverage requirement}. The objective is to find a solution $(C', F')$, where 
$F' \subseteq F$, $C' \subseteq C$, with $|C'| \ge p$, and
$$\cost(C', F') := \sum_{i \in F'} f_i + \sum_{j \in C'} d(j, F')$$
is minimized over all $(F', C')$, where $|C'| \ge p$.
									
\item[\textsc{FacLoc} with Penalties:] In addition to $F$, $C$, opening costs $\{f_i | i \in F\}$, and metric $d$, we are also given 
penalties $p_j \ge 0$ for each client $j \in C$. The objective is to find a solution $(C', F')$, where $F' \subseteq F$, and $C' \subseteq C$, 
such that, 
$$\cost(C', F') := \sum_{i \in F'} f_i +  \sum_{j \in C'} d(j, F') + \sum_{j \in C \setminus C'} p_j$$
is minimized over all $(C', F')$.
\end{description}

In this paper we present distributed $O(1)$-approximation algorithms for Robust \textsc{FacLoc} and \textsc{FacLoc} with Penalties
in several models of large-scale distributed computation. As far as we know, these are the first distributed algorithms for 
versions of \textsc{FacLoc} that are robust to outliers.
In distributed settings, the complexity of the problem can be quite sensitive to the manner in which input is specified.
We consider two alternate ways of specifying the input to the problem.
\begin{description}
\item[Explicit metric:] The metric $d$ is specified \textit{explicitly} as a $|F| \times |C|$ matrix distributed
among the machines of the underlying communication network.
This explicit description of the metric assumes that the $|F| \times |C|$ matrix fits
in the total memory of all machines combined.

\item[Implicit metric:] In this version, the metric is specified \textit{implicitly} -- as the shortest path metric of a given edge-weighted
	graph whose vertex set is $C \cup F$; we call this the \textit{metric graph}.
The reason for considering this alternate specification of the metric is that it can be quite compact; the graph specifying the metric can
be quite sparse (e.g., having $O(|F|+|C|)$ edges).
Thus, in settings where $|F| \cdot |C|$ is excessively large, but $|F| + |C|$ is not, this is a viable option.
\end{description}
For the facility location problems considered in this paper, when the input metric is explicitly specified, the biggest
challenge is solving the \textit{maximal independent set (MIS)} problem 
efficiently.
When the input metric is implicitly specified, the biggest challenge is to efficiently learn just enough of the metric
space.
Thus, changing the input specification changes the main challenge in a 
fundamental way and consequently we obtain very different results for the 
two alternate input specifications.

Our algorithms run in 3 models of distributed computation, which we now describe specifically in the context of facility location problems. 
All three models are synchronous message passing model.

\begin{description}
\item[Congested Clique model:] The Congested Clique model was introduced by Lotker et al.~\cite{lotker2005mstJournal} and 
then extensively studied in recent years \cite{GhaffariPODC17,GhaffariGKMRPODC18,GehweilerSPAA2006,Censor15podc,HolzerP14,JurdzinskiNSODA18,Nanongkai14,drucker2012task,DolevLP12,Patt-ShamirT11,Lenzen13}. 
In this model, the underlying communication network is a clique and the number of nodes in this clique equals $|F|+|C|$.
In each round, each node performs local computation based on the information it holds and then sends a (possibly distinct) $O(\log n)$-size 
message to each of the remaining nodes.
Initially, each node hosts a facility or a client and the node hosting facility $i$ knows the opening cost $f_i$ and the node hosting client $j$ knows the penalty $p_j$ for the \textsc{FacLoc} with penalties problem.
In the explicit metric setting, the node hosting facility $i$ knows all the connection costs $d(i, j)$ to all clients $j \in C$. 
Similarly, the node hosting client $j$ knows all connection costs $d(i, j)$ to all facilities $i \in F$.
In the implicit metric setting, the node hosting a facility or client knows the edges of the metric graph incident on
that facility or client. We call this input distribution \textit{vertex-centric} because each node is responsible for the local input of a facility or client. The vertex-centric assumption can be made without loss of generality because an adversarially (but evenly) distributed input can be redistributed in a vertex-centric manner among the nodes in constant rounds using Lenzen's routing protocol \cite{Lenzen13}.

\item[Massively Parallel Computation (MPC) model:] The MPC model was introduced in \cite{KarloffSVSODA2010} and variants of this model were considered in 
\cite{GoodrichSZISAAC11,BeameKSPODS13,AndoniNOYSTOC14}.
It can be viewed as a clean abstraction of the MapReduce model.
We are given $k$ machines, each with $S$ words of space and the input is distributed in a vertex-centric fashion among the machines, the only difference being that machines can host multiple facilities and clients (provided they fit in memory). Let $I$ be the total input size.
Typically, we require $k$ and $S$ to each be sublinear in $I$, that is $O(I^{1-\epsilon})$ for some $\epsilon > 0$.
We also require that the total memory not be too much larger than needed for the input, i.e., $k \times S = O(I)$.
In each round, each machine sends and receives a total of $O(S)$ words of information because it is the volume of information that will
fit into its memory. In our work we consider MPC algorithms with memory $S = \tilde{O}(n)$ \footnote{Throughout the paper, we use $\tilde{O}(f(n))$ as a shorthand for $O(f(n) \cdot \mbox{poly}(\log n))$ and $\tilde{\Omega}(f(n))$ as a shorthand for $\Omega(f(n)/\mbox{poly}(\log n))$.} where $n = |F| = |C|$. In the explicit metric setting, since \(I = O(n^2)\), even if we assume $S = \tilde{O}(n)$, \(k\) and \(S\) are still strictly sublinear in \(I\). But in the implicit metric setting, if we assume $S = \tilde{O}(n)$ then the memory may not be strictly sublinear in the input size when the input graph is sparse, having \(O(n)\) edges for example. Therefore, our algorithms are not strictly MPC algorithms when the input is sparse. Similar to the Congested Clique model, we can assume that the input is distributed in a vertex-centric manner without loss of generality, due to the nature of communication in each round and the fact that $S = \Omega(n)$.

\item[$k$-machine model:] The $k$-machine model was introduced in \cite{KlauckNPRSODA15} and further studied in \cite{PanduranganRSSPAA18}. 
This model abstracts essential features of systems such as Pregel 
\cite{MalewiczABDHLCSIGMOD2010} 
and Giraph (see \verb+http://giraph.apache.org/+) that have been designed
for large-scale graph processing.
We are given $k$ machines and the input is distributed among the machines. In \cite{KlauckNPRSODA15}, the $k$-machine model is used to solve graph problems
and they assume a \textit{random vertex partition} distribution of the input graph among the $k$ machines. In other words, each vertex along
with its incident edges is provided to one of the $k$ machines chosen uniformly at random. 
The corresponding assumption for facility location problems would be that each facility and each client is assigned uniformly at random
to one of the $k$ machines.
Facility $i \in F$ comes with its opening cost $f_i$ and client $j \in C$ comes with its penalty $p_j$ for the \textsc{FacLoc} with
penalties problem.
In the explicit metric setting, each facility $i \in F$ comes with connections costs $d(i, j)$ for all $j \in C$ whereas
in the implicit metric setting facility $i$ comes along with the edges of the metric graph incident on it.
Similarly for each client $j \in C$.
In each round, each machine can send a (possibly distinct) size-$B$ message to each of the remaining $k-1$ machines. Typically, $B$ is assumed to
be $\text{poly}(\log n)$ bits \cite{KlauckNPRSODA15}.
\end{description}

The Congested Clique model does not directly model settings of large-scale computation because in this model
the number of nodes in the underlying communication network equals the number of vertices in the input graph.
However, fast Congested Clique algorithms can usually be translated (sometimes automatically) to fast MPC 
and $k$-machine algorithms. So the Congested Clique algorithms in this paper 
are important stepping stones
towards more complex MPC and $k$-machine algorithms \cite{HegemanPTCS15, KlauckNPRSODA15}.
The MPC model and the $k$-machine model are quite similar. 
Even though the $k$-machine model is specified with a per-edge bandwidth constraint of $B$ bits, it can be equivalently described
with a per-machine bandwidth constraint of $k \cdot B$ bits that can be sent and received in each round. 
Thus setting $k \cdot B = S$ makes the $k$-machine model and MPC model equivalent in their bandwidth constraint.
Despite their similarities, it is useful to think about both models due to differences in how they are parameterized
and how these parameters affect the running times of algorithms in these models.
For example, in the MPC model, usually one starts by picking $S$ as a sublinear function of the input size $n$.
This leads to the number of machines being fixed and the running time of the algorithm is expressed as a function of $n$.
In the $k$-machine model $B$ is usually fixed at $\text{poly}(\log n)$ and the running time of the algorithm is
expressed as a function of $n$ and $k$. 
This helps us understand how the running time changes as we increase $k$. For example,
algorithms with running times of the form $O(n/k)$ exhibit a linear speedup as $k$ increases, whereas algorithms with running time 
of the form $O(n/k^2)$ indicating a quadratic speedup \cite{PanduranganRSSPAA2016}.

\subsection{Main Results}

In order to obtain $O(1)$-approximation algorithms for Robust \textsc{FacLoc} and \textsc{FacLoc} with Penalties,
Charikar et al.~\cite{CharikarSODA2001} propose modifications to the primal-dual approximation algorithm for \textsc{FacLoc} due to 
Jain and Vazirani \cite{JainVaziraniJACM2001}. 
The problem with using this approach for our purposes is that it seems difficult obtain fast \textit{distributed}
algorithms using the Jain-Vazirani approach. For example, obtaining a 
\textit{sublogarithmic} round $O(1)$-approximation for \textsc{FacLoc} in the 
Congested Clique model using this approach seems difficult.
However, as established in our previous work \cite{HegemanPDC15,BandyapadhyayArxiv2017,Berns2012} and in \cite{GarimellaDGSCIKM2015} the greedy algorithm of Mettu 
and Plaxton \cite{MettuPlaxtonSICOMP2003} for \textsc{FacLoc} seems naturally suited for fast distributed implementation.

The first contribution of this paper is to show that $O(1)$-approximation
algorithms to Robust \textsc{FacLoc} and \textsc{FacLoc} with Penalties
can \textit{also} be obtained by using variants of the Mettu-Plaxton greedy
algorithm.
Our second contribution is to show that by combining ideas from earlier
work \cite{HegemanPDC15,BandyapadhyayArxiv2017} with some new ideas, we can efficiently implement 
distributed versions of the variants of the Mettu-Plaxton algorithm 
for Robust \textsc{FacLoc} and \textsc{FacLoc} with Penalties.
The specific results we obtain for the two versions of input specification are as follows.
For simplicity of exposition, we assume $|C| = |F| = n$.

\begin{itemize}
\item \textbf{Implicit metric}: For both problems, we present $O(1)$-approximation
algorithms running in $O(\mbox{poly}(\log n))$ rounds in the Congested Clique
and the MPC model.
Assuming the metric graph has $m$ edges, the input size is $\Theta(m+n)$
and we use $\tilde{O}(m/n)$ machines each with memory $\tilde{O}(n)$.
In the $k$-machine model, we present $O(1)$-approximation algorithms running in
$\tilde{O}(n/k)$ rounds.

\item \textbf{Explicit metric}: For both problems, we present extremely fast 
$O(1)$-approximation algorithms, running in $O(\log\log\log n)$ rounds, 
in the Congested Clique and the MPC model. The
input size is $\Theta(n^2)$ and we use $n$ machines each with memory $\tilde{O}(n)$
in the MPC model.
In the $k$-machine model, we present $O(1)$-approximation algorithms running in $\tilde{O}(n/k)$ rounds.
\end{itemize}

\ifthenelse{\boolean{short}}{
  Due to space constraints we only describe our distributed implementations for the Robust \textsc{FacLoc} algorithm and omit most of the technical proofs. The full version with all the technical details appears in \cite{InamdarPPArxiv2018}.
}{} \section{Sequential Algorithms for Facility Location with Outliers}
We first describe the greedy sequential algorithm of Mettu and Plaxton \cite{MettuPlaxtonSICOMP2003} (Algorithm \ref{alg:MP}) for the Metric \textsc{FacLoc} problem which will serve as a building block for our algorithms for Robust \textsc{FacLoc} and the \textsc{FacLoc} with Penalties discussed in this section.
The algorithm first computes a ``radius'' $r_i$ for each facility $i \in F$ and it then greedily picks facilities to open in non-decreasing order
of radii provided no previously opened facility is too close.
The ``radius'' of a facility $i$ is the amount that each client is charged for the opening of facility $i$.
Clients pay towards this charge after paying towards the cost of connecting to facility $i$; clients that have a large connection cost to $i$
pay nothing towards this charge.
\RestyleAlgo{boxruled}
\begin{algorithm2e}\caption{\textsc{FacilityLocationMP}\((F, C)\)}\label{alg:MP}
	\tcc{Radius Computation Phase:}
	For each $i \in F$, compute $r_i \ge 0$, satisfying $f_i = \sum_{j \in C} \max\{0, r_i - c_{ij}\}$. \\
	\tcc{Greedy Phase:}
	Sort and renumber facilities in the non-decreasing order of $r_i$. \\
	$F' \gets \emptyset$ \Comment{Solution set}\\
	\For{$i = 1, 2, \ldots$} { 
		\If{there is no facility in $F'$ within distance $2r_i$ from $i$} {
			$F' \gets F' \cup \{i\}$
		}
	}
	Connect each client $j$ to its closest facility in $F'$.
\end{algorithm2e}
It is shown in \cite{MettuPlaxtonSICOMP2003} using a charging argument that Algorithm \ref{alg:MP} is 3-approximation for the Metric \textsc{FacLoc} problem. Later on, \cite{ArcherRSESA2003} gave a primal-dual analysis, showing the same approximation guarantee, by comparing the cost of the solution to a dual feasible solution. We use the latter analysis approach as it can be easily modified to work for the algorithms with outliers. 

\ifthenelse{\boolean{short}}{}{
  For a facility $i \in F$ and a client $j \in C$, we use the shorthand $c_{ij} \coloneqq d(i, j)$. Also, for a facility $i \in F$ and a radius $r \ge 0$, let $B(i, r)$ denote the set of clients within the distance $r$, i.e., $B(i, r) \coloneqq \{j \in C\mid c_{ij} \le r \}$.
}

\subsection{Robust Facility Location}
\label{sec:robust-fl-seq}
Since we use the primal dual analysis of \cite{ArcherRSESA2003} to get a bounded approximation factor, we need to address the fact 
that the standard linear programming relaxation for Robust \textsc{FacLoc} has unbounded integrality gap. 
To fix this we modify the instance in a similar manner to \cite{CharikarSODA2001}.
Let $(C^*, F^*)$ be a fixed optimal solution, and let $i_* \in F$ be a facility in that solution with the maximum opening cost $f_{i_*}$. 
We begin by assuming that we are given a facility, say $i_e$ with opening cost $f_{i_e}$, such that, $f_{i_*} \le f_{i_e} \le \alpha f_{i_*}$, 
where $\alpha \ge 1$ is a constant. 
Now, we modify the original instance by changing the opening costs of the facilities as follows. 

$$f'_i = \begin{cases}
+\infty &\text{ if $f_i > f_{i_e}$}
\\0 &\text{ if $i = i_e$}
\\f_i &\text{ otherwise}
\end{cases}$$

Note that we can remove the facilities with opening cost $+\infty$ without affecting the cost of an optimal solution, and hence we assume that w.l.o.g. all the modified opening costs $f'_i$ are finite.

Let $(C^*_e, F^*_e)$ be an optimal solution for this modified instance, and let $\cost_e(C^*_e, F^*_e)$ be its cost using the modified opening costs. Observe that without loss of generality, we can assume that $i_e \in F^*_e$, since its opening cost $f'_{i_e}$ equals $0$.
We obtain the following lemma and its simple corollary.

\begin{lemma} \label{lem:preproc}
	$\cost_e(C^*_e, F^*_e)  \le \cost(C^*, F^*)$.
\end{lemma} 
\ifthenelse{\boolean{short}}{}{
  \begin{proof}	
    Recall that $i_e$ satisfies $f_{i_*} \le f_{i_e} \le \alpha f_{i_*}$ where $i_*$ is the facility with largest opening cost in $F^*$. In the modified instance, all facilities with opening cost greater than $f_{i_e}$ are removed, however, no facility from $F^*$ is removed, because $f_{i_*} \le f_{i_e}$. Therefore, $(C^*, F^*)$ is a feasible solution for the modified instance. This implies $\cost_e(C^*_e, F^*_e) \le \cost_e(C^*, F^*)$, which follows from the optimality of $(C^*_e, F^*_e)$. Finally, recall that for any facility $i \in F^*$, $f'_i = f_i$, and hence $\cost_e(C^*, F^*) = \cost(C^*, F^*)$.
  \end{proof}
}

\begin{corollary} \label{cor:modified-cost}
  Let $(C'_e, F'_e)$ be a feasible solution for the instance with modified facility opening costs, such that, $\cost_e(C'_e, F'_e) \le \beta \cdot \cost_e(C^*_e, F^*_e) + \gamma \cdot f_{i_e}$ (where $\beta \ge 1, \gamma \ge 0$). Then, $(C'_e, F'_e)$ is a $\beta + \alpha \cdot (\gamma + 1)$ approximation for the original instance.
\end{corollary}
\ifthenelse{\boolean{short}}{}{
  \begin{proof}
    Consider,
    \begingroup
    \allowdisplaybreaks
    \begin{align*}
      \cost(C'_e, F'_e) &\le \cost_e(C'_e, F'_e) + f_{i_e} \tag{if $i_e \in F'_e$} 
      \\&\le \beta \cdot \cost_e(C^*_e, F^*_e) + (\gamma+1) \cdot f_{i_e} 
      \\&\le \beta \cdot \cost(C^*, F^*) + (\gamma + 1) \cdot f_{i_e} \tag{From lemma \ref{lem:preproc}}
      \\&\le \beta \cdot \cost(C^*, F^*) + \alpha \cdot (\gamma + 1) \cdot f_{i_*} \tag{$f_{i_e} \le \alpha f_{i_*}$}
      \\&\le \beta \cdot \cost(C^*, F^*) + \alpha \cdot (\gamma + 1) \cdot \cost(C^*, F^*) \tag{$f_{i_*} \le \cost(C^*, F^*)$}
    \end{align*}
    \endgroup
  \end{proof}
}

To efficiently find a facility $i_e$ satisfying $f_{i_*} \le f_{i_e} \le \alpha f_{i_*}$, 
we partition the facilities into sets where each set contains facilities with opening costs from the range $\left[(1+\epsilon)^i, (1+\epsilon)^{i+1}\right)$. Iterating over all such ranges, and choosing a facility with highest opening cost from that range, we are guaranteed to find a facility $i_e$ such that, $f_{i_*} \le f_{i_e} \le (1 + \epsilon) f_{i_*}$ \footnote{An alternative approach would be to consider each facility one-by-one as a candidate, but for an efficient distributed implementation we can only afford $O(\log n)$ distinct guesses.}. The total number of such iterations will be \(O(\log_{1+\epsilon} \frac{f_{\max}}{f_{\min}})\), where $f_{\max}$ is the largest opening cost, and $f_{\min}$ is the smallest non-zero opening cost. Assuming that every individual item in the input (e.g., facility opening costs,
connection costs, etc.) can each be represented in $O(\log n)$ bits and that $\epsilon$ is a constant, this amounts to $O(\log n)$ iterations.

Our facility location algorithm is described in Algorithm \ref{alg:outlierSeq}. This algorithm can be thought of as running $O(\log n)$ separate instances of a modified version of the original Mettu-Plaxton algorithm (Algorithm \ref{alg:MP}), where in each instance of the Mettu-Plaxton algorithm, the algorithm is terminated as soon as the number of outlier clients drops below the required number, following which there is some post-processing.

\RestyleAlgo{boxruled}
\begin{algorithm2e}\caption{\textsc{RobustFacLoc}\((F, C, p)\)}\label{alg:outlierSeq}
  \tcc{Recall: $\ell \coloneqq |C| - p$}
  
  \For{$t = 0, \cdots, O(\log n)$} {
    Let $i_e \in F$ be the most expensive facility from the facilities with opening costs in the range $[(1 + \epsilon)^t, (1 + \epsilon)^{t+1})$ for some small constant $\epsilon>0$ \label{lin:i-e_chosen}\\
    Modify the facility opening costs to be
    $$f'_i = \begin{cases}
      +\infty &\text{ if $f_i > f_{i_e}$}
      \\0 &\text{ if $i = i_e$}
      \\f_i &\text{ otherwise}
    \end{cases}$$ \\
    \tcc{Radius Computation Phase:}
    For each $i \in F$, compute $r_i \ge 0$, satisfying $f'_i = \sum_{j \in C} \max\LR{0, r_i - c_{ij}}$. \\
    \tcc{Greedy Phase:}
    Sort and renumber facilities in the non-decreasing order of $r_i$. \\
    Let $C' \gets \emptyset$, $F' \gets \emptyset$, $O' \gets C$ \\
    Let $F_0 \gets \emptyset$ \\
    \For{$i = 1, 2, \ldots$} { \label{alg1:for}
      \If{there is no facility in $F'$ within distance $2r_i$ from $i$} {
        $F' \gets F' \cup \{i\}$
      }
      $F_i \gets F_{i-1} \cup \{i\}$ \\
      Let $C_i$ denote the set of clients that are within distance \(r_i\) \label{alg1:remove_clients}\\
      $C' \gets C' \cup  C_i$, \quad $O' \gets O' \setminus  C_i$. \\
      \textbf{if} $|O'| \le \ell$ \textbf{then break} \label{alg1:break} \\
    }\label{alg1:endfor}
    \tcc{Outlier Determination Phase:}
    \If{$|O'| > \ell$} {
      Let $O_1 \subseteq O'$ be a set of $|O'| - \ell$ clients that are closest to facilities in $F'$. \\
      $C' \gets C' \cup O_1$, \quad $O' \gets O' \setminus O_1$.
    }
    \ElseIf{$|O'| < \ell$} {
      Let $O_2 \subseteq C'$ be the set of $\ell - |O'|$ clients with largest distance to open facilities $F'$. \\
      $C' \gets C' \setminus O_2$, \quad $O' \gets O' \cup O_2$.
    }
    Let $(C'_t, F'_t) \gets (C', F')$ 
  }

  \Return Return $(C'_t, F'_t)$ with the minimum cost.
\end{algorithm2e}
We abuse the notation slightly, and denote by $(C', F')$ the solution returned by the algorithm, i.e., the solution $(C'_t, F'_t)$ corresponding to the iteration $t$ of the outer loop that results in a minimum cost solution. Similarly, we denote by $i_e$ a facility chosen in line \ref{lin:i-e_chosen} in the iteration corresponding to this iteration $t$.

\ifthenelse{\boolean{short}}{}{
Moreover, we will consider the facility costs $f'$ to be the modified facility opening costs in the same iteration $t$. We can ignore the facilities with opening cost $+\infty$ and add $i_e$ to any solution with no additional cost. Therefore in our analysis, we just ignore these facilities and use the original facility costs $f$ for other facilities since they are the same as the modified costs.

Note that we exit the greedy phase if we either process all facilities or we break at line \ref{alg1:break}, each of which corresponds to the cases -- (i) $|O'| > \ell$ where some outliers become clients and (ii) $|O'| < \ell$ where some clients become outliers again, in the outlier determination phase (we are done if $|O'| = \ell$).

  Let \(C''\) and \(O''\) denote the sets \(C'\) and \(O'\) just before the outlier determination phase. Note that we we exit the greedy phase if we either process all facilities or we break at line \ref{alg1:break}, each of which corresponds to the cases -- (i) $|O| > \ell$ and (ii) $|O| < \ell$ in the outlier determination phase (we are done if $|O| = \ell$).
  
  For a client $j \in C$, let $v'_j \coloneqq \min_{i \in F} \max\{r_i, c_{ij}\}$. To make the analysis easier, we consider a more expensive solution $(\tilde{C}', F')$ where the set of clients $\tilde{C}'$ is constructed using the following modified outlier determination phase:
  
  \begin{enumerate}
  \item If $|O''| > \ell$, then let $O_1 \subseteq O''$ be a set of $|O''| - \ell$ clients that have the smallest $v'_j$-values. In this case, we let $\tilde{C'} \gets C'' \cup O_1$, \quad $\tilde{O}' \gets O'' \setminus O_1$.
    
  \item Otherwise, if $|O''| < \ell$. Let $O_2 \subseteq C_i$ be the set of clients with largest $\ell - |O''|$ $v'_j$-values ($i$ is last iteration). In this case, we let $\tilde{C}' \gets C'' \setminus O_2$, \quad $\tilde{O}' \gets O'' \cup O_2$.
  \end{enumerate}
  
  It is easy to see by an exchange argument that the $\cost_e(C', F') \le \cost_e(\tilde{C}', F')$, the outliers determined in the algorithm are at least as far from \(F'\) as ones in the modified outlier determination phase. Henceforth, we analyze the cost of the solution $(\tilde{C}', F')$ by comparing it to the cost of a feasible dual \LP solution and in order to alleviate excessive notation, we will henceforth refer to the solution $(\tilde{C}', F')$ as $(C', F')$ and $\tilde{O}'$ as $O'$. We state the standard primal and dual linear programming relaxations for the Robust \textsc{FacLoc} problem in Figure \ref{fig:outlier-primal-dual}
  
  \begin{figure}
    \centering
    \makebox[\textwidth][c]{
      
      \begin{minipage}{0.65\textwidth}
        \begin{mdframed}[backgroundcolor=gray!9]
          \fontsize{9}{9.4}\selectfont
          \begin{alignat}{3}
            \text{minimize}   \displaystyle&\sum\limits_{i \in F} f_i y_i + \sum_{i \in F, j \in C} c_{ij} x_{ij} & \nonumber\\
            \text{subject to } \displaystyle& z_j + \sum\limits_{i \in F} x_{ij} \geq 1,  \quad &\forall j \in C \label{constr:fllp-coverage} \\
            \displaystyle&\qquad\quad\ \   x_{ij} \le y_i, &\forall  i \in F, \forall j \in C  \label{constr:fllp-open}\\
            \displaystyle&\qquad\ \sum_{j \in C} z_j \le \ell, \label{constr:fllp-outlier}\\
            \displaystyle&\quad r_j, y_i, x_{ij} \in [0, 1] &\forall i \in F,\forall j \in C \label{constr:fllp-nonneg}
          \end{alignat}
          \centering Primal LP
        \end{mdframed}
      \end{minipage}
      \makebox[0.65\textwidth][c]{
        \begin{minipage}{0.63\textwidth}
          \begin{mdframed}[backgroundcolor=gray!9]
            \fontsize{9}{10}\selectfont
            \begin{alignat}{3}
              \text{maximize}\displaystyle&\qquad \sum\limits_{j \in C} v_j - \ell \cdot \q & \nonumber\\
              \text{subject to } \displaystyle&\qquad v_j \le c_{ij} + w_{ij},  \quad &\forall j \in C \label{constr:flld-alpha} \\
              \displaystyle& \sum_{j \in C} w_{ij} \le f_i &\forall i \in F \label{constr:fllpd-beta}\\
              \displaystyle&\qquad\ v_j \le \q &\forall j \in C \label{constr:fllpd-alphaub}\\
              \displaystyle&\quad v_j, w_{ij} \ge 0 &\forall i \in F,\forall j \in C \label{constr:fllp-integral}
              \\\nonumber
            \end{alignat}
            \centering Dual LP
          \end{mdframed}
        \end{minipage}}
    }
    \caption{Primal and Dual Linear Programming Relaxations for Robust \textsc{FacLoc} \label{fig:outlier-primal-dual}}
  \end{figure}
  
  Now, we construct a feasible dual solution $(v, w, \q)$.
  
  For a facility $i \in F$ and client $j \in C$, let  $w_{ij} \coloneqq \max\{0, r_i - c_{ij}\}$. Let $\q \coloneqq \max_{j \in C'} v'_{j}$ (recall that  $v'_j \coloneqq \min_{i \in F} \max\{r_i, c_{ij}\} = \min_{i \in F} {c_{ij} + w_{ij}}$). Now, for a client $j \in C$, define $v_j$ as follows:
  $$v_j = \begin{cases}
    v'_j &\text{if } j \in C'
    \\\q &\text{if } j \in O'
  \end{cases}$$
  
  \begin{claim} \label{claim:bound_vj}
    A client \(j \in C_i\) iff \(v_j' \le r_i\)
  \end{claim}
  \begin{proof}
    Client \(j\) is added to \(C_i\) iff for some \(i' \in F_i\), \(j \in B(i', r_i)\). This means \(c_{i'j} \le r_i\) and \(r_{i'} \le r_i\) since we process facilities in increasing order of \(r\)-value. This means \(v'_j \le \max\{r_{i'}, c_{i'j}\} \le r_i\)
  \end{proof}
  
  \begin{lemma}
    The solution $(v, w, \q)$ is a feasible solution to the dual LP relaxation \ref{fig:outlier-primal-dual}.
  \end{lemma}
  \begin{proof}
    Note that constraints \ref{constr:fllpd-beta}, \ref{constr:fllpd-alphaub}, and \ref{constr:fllp-integral} of the dual are satisfied by construction and so is constraint \ref{constr:flld-alpha} for clients $j \in C'$. Therefore, in order to show that the solution $(v, w, \q)$ is feasible, we have to show that constraint \ref{constr:flld-alpha} is satisfied for all clients $j \in O'$. To this end, we consider the following two cases.
    
    \textbf{Case 1.} We enter the outlier determination phase after iterating over all facilities in $F$. Therefore, we have $|O''| > \ell$. This means that we identified a set \(O_1 \subset O''\) of size \(|O''| - \ell\) to be marked as non-outliers.
    
    As we iterate over all the facilities, if $j \in O''$ then by claim \ref{claim:bound_vj}, $v'_j > \max_{i \in F}{r_i} \ge \max_{j \in C''}{v'_j}$.
    
    We put in \(O_1\) the clients from \(O''\) that have smallest \(v'_j\)-values. This means that for all \(j \in O'\), \(v'_j \ge \max_{j \in O_1}{v'_j} = \max_{j \in C'' \cup O_1}{v'_j} = \max_{j \in C'}{v'_j} = \q\) where the second equality is because for \(v'_j\) for \(j \in C''\) is at most \(v'_{j'}\) for \(j' \in O_1\). 
    
    Therefore, we conclude that for any $j \in O'$ and $i \in F$, $c_{ij} + w_{ij} \ge v'_j \ge \q = v_j$.
    
    \textbf{Case 2.} We enter the outlier determination phase because of the break statement on line \ref{alg1:break}. Here, $|O''| \le \ell$ and \(C' \subseteq C''\)
    
    Let \(i^*\) be the last iteration of the for loop. Therefore \(F_{i^*}\) is the set of facilities we consider in the for loop. 
    
    Recall that by the case assumption we have $|O''| \le \ell$ and therefore $C' \subseteq C''$. All clients \(j \in C''\) were part of \(C_i\) for some \(i \in F_{i^*}\) and by claim \ref{claim:bound_vj} we have \(v'_j \le r_i \le r_{i^*}\). Therefore,
    $$\q = \max_{j \in C'} v_j \le \max_{j \in C''} v_j \le r_{i^*}.$$
    
    Let $j \in O'$ be an outlier client. If $j \in O''$, then for any facility $i \in F$,
    \begin{align*}
      c_{ij} + w_{ij} &\ge v'_j \\
                      &\ge r_{i^*} \tag{Otherwise $j$ would be in $C''$ by claim \ref{claim:bound_vj}} \\
                      &\ge \q
                        = v_j
    \end{align*}
    
    Otherwise, $j \in O_2$ and was added to $O'$ because it had highest $v'_j$ value among $C_{i^*}$. Therefore, it follows that for any facility $i \in F$, $c_{ij} + w_{ij} \ge v'_j \ge \max_{j \in C'} v'_{j} = \q = v_j$
    
    From the case analysis it follows that $v_j \le c_{ij} + w_{ij}$ for all $j \in O'$ and for all $i \in F$. Therefore, we have shown that $(v, w, \q)$ is a dual feasible solution.
  \end{proof}
  
  For the approximation guarantee we can focus just on the clients in $C'$ because the only contribution that the clients in $O'$ make to the dual objective function is to cancel out the $- \ell \q$ term and hence they do not affect the approximation guarantee. We call a facility $\i$ the \emph{bottleneck} of $j$ if \(v'_j = \max\{r_{\i}, c_{\i j}\} =  c_{\i j} + w_{\i j}\). We first prove a few straightforward claims about the dual solution.
  
  \begin{claim} \label{cl:apx-1}
    For any $i \in F$ and $j \in C$, $r_i \le c_{ij} + w_{ij}$. Moreover if $w_{ij} > 0$ then $r_i = c_{ij} + w_{ij}$
  \end{claim}
  \begin{proof}
    $w_{ij} = \max\{0, r_i - c_{ij}\} \ge r_i - c_{ij}$. Now if $w_{ij} > 0$ then $r_i > c_{ij}$ and we have $w_{ij} = r_i - c_{ij}$ which implies the claim.
  \end{proof}
  
  \begin{claim} \label{cl:apx-2}
    If $\i$ is a bottleneck for $j \in C'$, then $v_j \ge r_{\i}$.
  \end{claim}
  \begin{proof}
    For $j \in C'$, $v_j = v'_j = \max\LR{c_{\i j}, r_{\i}} \ge r_{\i}$.
  \end{proof}
  
  \begin{claim} \label{cl:apx-3}
    If $\i \in F'$ is a bottleneck for $j \in C'$, then $w_{i'j} = 0$ for all $i' \in F'$, where $i' \neq \i$.
  \end{claim}
  \begin{proof}
    Assume for contradiction that $w_{i'j} > 0$, i.e., $r_{i'} \ge c_{i'j}$.
    
    If $r_{\i} \ge c_{\i j}$, then $v_j = r_{\i} \le c_{i'j} + w_{i'j} = r_{i'}$. In this case, $c_{\i i'} \le c_{\i j} + c_{i' j} \le 2 r_{i'}$.
    
    Otherwise, if $c_{\i j} > r_{\i}$, then $v_j = c_{\i j} \le r_{i'}$. Here too we have, $c_{\i i'} \le c_{\i j} + c_{i' j} \le 2r_{i'}$.
    
    In either case, $c_{\i i'} \le 2r_{i'} \le 2 \max\LR{r_{i'}, r_{\i}}$, which is a contradiction, since at most one of $i', \i$ can be added to $F'$.
  \end{proof}
  
  \begin{claim} \label{cl:apx-4}
    If a closed facility $\i \not\in F'$ is the bottleneck for $j \in C'$, and if an open facility $i' \in F'$ caused $\i$ to close, then $c_{i' j} \le 3 v_j$.
  \end{claim}
  \begin{proof}
    $\i$ is a bottleneck for $j$, so $v_j \ge c_{\i j}$, and $v_j \ge r_{\i}$.
    
    Since $i' \in F'$ caused $\i$ to close, then $c_{k i} \le 2r_{\i} \le 2 v_j$.
    Therefore, $c_{i' j} \le c_{k i} + c_{i j} \le 3 v_j$.
  \end{proof}
  
  Now we state the main lemma that uses the dual variables for analyzing the cost.
  
  \begin{lemma} \label{lem:outlier-lemma}
    For any $j \in C'$, there is some $i \in F'$ such that $c_{ij} + w_{ij} \le 3 v_j$.
  \end{lemma}
  \begin{proof}
    Fix a client $j \in C'$, and let $\i$ be its bottleneck. We consider different cases. Here, $c_{ij}$ should be seen as the connection cost of $j$, and $w_{ij}$, the cost towards opening of a facility (if $w_{ij} > 0$).
    
    \begin{itemize}
    \item \textbf{Case 1:} $\i \in F'$.
      
      In this case, by claim \ref{cl:apx-3}, $w_{i'j} = 0$ for all other $i' \in F'$.
      
      If $c_{\i j} < r_{\i}$, then $w_{\i j} > 0$, and $v_j = c_{\i j} + w_{\i j}$. Therefore, $v_j$ pays for the connection cost as well as towards the opening cost of $\i$.
      
      Otherwise, if $c_{\i j} \ge r_{\i}$, then $w_{\i j} = 0$. Also, $w_{i'j} = 0$ for all other $i' \in F'$. Therefore, $j$ does not contribute towards opening of any facility in $F'$. Also, we have $v_j = \max\LR{c_{\i j}, r_{\i}} = c_{\i j}$, i.e., $v_j$ pays for $j$'s connection cost.
      
    \item \textbf{Case 2:} $\i \notin F'$ and $w_{ij} = 0$ for all $i \in F'$.
      
      Let $i' \in F'$ be the facility that caused $\i$ to close. From claim \ref{cl:apx-4}, we have that $c_{i' j} \le 3 v_j$, i.e. $3 v_j$ pays for the connection cost of $j$.
      
    \item \textbf{Case 3:} $\i \notin F'$, and there is some $i' \in F'$ with $w_{i'j} > 0$. But $i'$ did not cause $\i$ to close.
      
      Since $w_{i'j} > 0$, by claim \ref{cl:apx-1} $w_{i'j} = r_{i'} - c_{i'j}$ and $r_{i'} > c_{i'j}$.
      
      Let $i$ be the facility that caused $\i$ to close. Therefore, $c_{\i i} \le 2r_{\i}$. Also, $c_{i j} \le c_{i \i} + c_{\i j} \le 2r_{\i} + c_{\i j} \le 3v_j$.
      
      Now, since $i$ and $i'$ both belong to $F'$, $c_{i' i} > 2\max \LR{r_i, r_{i'}} \ge 2r_{i'} = 2(c_{i' j} + w_{i'j})$.
      
      Therefore, by triangle inequality, $c_{ij} + c_{i'j} \ge c_{i'i} > 2(c_{i' j} + w_{i'j})$ which implies $c_{ij} > c_{i'j} + 2w_{i'j}$.
      
      It follows that, $c_{i' j} + w_{i'j} \le c_{i' j} + 2w_{i'j} < c_{i j} \le 3 v_j$. Therefore, $3v_j$ pays for the connection cost of $j$ to $i'$, and its contribution towards opening of $i'$.
      
    \item \textbf{Case 4:} $\i \not\in F'$, but for $i' \in F'$ that caused $\i$ to close has $w_{i'j} > 0$.
      
      Again, since $w_{i'j} > 0$, by claim \ref{cl:apx-1} $w_{i'j} = r_{i'} - c_{i'j}$ and $r_{i'} > c_{i'j}$.
      
      From claim \ref{cl:apx-2}, we have that $v_j \ge r_{\i}$. Then, since $i'$ caused $\i$ to close, we have $r_{\i} \ge r_{i'} = c_{i'j} + w_{i'j}$. This implies $v_j \ge c_{i'j} + w_{i'j}$, i.e., $v_j$ pays for the connection cost of $j$ to $i'$, and its contribution towards opening of $i'$.
    \end{itemize}  
  \end{proof}

  Now we are ready to prove the approximation guarantee of the algorithm.
}
\begin{theorem} \label{thm:outlier-theorem}
  $\cost_e(C', F') \le 3 \cdot \cost_e(C^*_e, F^*_e) + f_{i_e}$
\end{theorem}
\ifthenelse{\boolean{short}}{}{
  \begin{proof}
    Recall that $f_{i_e}$ denotes the cost of the most expensive facility in an optimal solution. Furthermore, notice that for any facility $i \in F' \setminus \{i^*\}$, the clients in the ball $B(i, r_i) \subseteq C'$. However, if $i^* \in F'$, some of the clients in $B(i^*, r_{i^*})$ may have been removed in the outlier determination phase, and therefore it may not get paid completely by the dual variables $v_j$. Therefore,
    
    \begingroup
    \allowdisplaybreaks
    \begin{align*}
      \cost_e(C', F') &= \sum_{j \in C'} d(j, F) + \sum_{i \in F' \setminus \{i^*\} } f_i + f_{i^*}
      \\&\le 3 \cdot \sum_{j \in C'} v_j + f_{i^*} \tag{From lemma \ref{lem:outlier-lemma}}
      \\&= 3 \cdot \lr{\sum_{j \in C} v_j - \q\ell} + f_{i^*} \tag{For $j \in O'$, $v_j = \q$, and $|O'| = \ell$}
      \\&\le 3 \cdot \lr{\sum_{j \in C} v_j - \q\ell} + f_{i_e} \tag{Since $i_e$ is the most expensive facility}
    \end{align*}
    \endgroup
    Since $(v, w, \q)$ is a feasible dual solution, its cost is a lower bound on the cost of any integral optimal solution. Therefore, the theorem follows.
  \end{proof}
}

\noindent
Applying Corollary \ref{cor:modified-cost} with $\alpha = 1 + \epsilon, \beta = 3, \gamma = 1$ yields the following approximation guarantee.
\begin{theorem}
  The solution returned by Algorithm \ref{alg:outlierSeq} is a \(5 + \epsilon\) approximation to the Robust \textsc{FacLoc} problem.
\end{theorem}

\subsection{Facility Location with Penalties} \label{sec:seq-fl-wp}
For the penalty version, each client $j$ comes with a penalty $p_j$ which is the cost we pay if we make $j$ an outlier. Therefore, the radius computation for a facility changes because if a facility $i$ is asking client $j$ to contribute more than $p_j - c_{ij}$ then it is cheaper for $j$ to mark itself as an outlier and pay its penalty. Therefore, for each facility $i \in F$, let $r_i \ge 0$ be a value such that $f_i = \sum_{j \in C} \max\LR{\min \LR{r_i - c_{ij}, p_j - c_{ij}}, 0}$, if it exists. Notice that if for a facility $i \in F$, such an $r_i$ does not exist, then it must be the case that for all $j \in C$, $p_j \le c_{ij}$. That is, it is for any client, it is cheaper to pay the penalty than to connect it to this facility. Therefore, removing such a facility from consideration does not affect the cost of any solution, and hence we assume that for all $i \in F$, an $r_i \ge 0$ exists such that $f_i = \sum_{j \in C} \max\LR{\min \LR{r_i - c_{ij}, p_j - c_{ij}}, 0}$.
The algorithm for \textsc{FacLoc} with Penalties is shown in Algorithm \ref{alg:penaltiesSeq}.

\begin{algorithm2e}\caption{\textsc{PenaltyFacLoc}\((F, C, p)\)}\label{alg:penaltiesSeq}
  \tcc{Radius Computation Phase:}
  Compute $r_i$ for each $i \in F$ satisfying $f_i = \sum_{j \in C} \max\LR{\min\LR{r_i - c_{ij}, p_j - c_{ij}}, 0}.$ \\
  \tcc{Greedy Phase:}
  Sort and renumber facilities in the non-decreasing order of $r_i$. \\
  $C' \gets \emptyset$, \quad $F' \gets \emptyset$, \quad $O' \gets \emptyset$. \\
  \For{$i = 1, 2, \ldots$} {
    \If{there is no facility in $F'$ within distance $2r_i$ from $i$} {
      $F' \gets F' \cup \{i\}$
    }
  }
  \tcc{Outlier Determination Phase:}
  \For{each client $j$} {
    Let $i$ be the closest facility to $j$ in $F'$ \\
    \textbf{if} $c_{ij} \le p_j$ \textbf{then} $C' \gets C' \cup \{j\}$ \\
    \textbf{else} $O' \gets O' \cup \{j\}$
  }
  \Return $(C', F')$ as the solution.
\end{algorithm2e}

\ifthenelse{\boolean{short}}{}{
  We state the standard primal and dual linear programming relaxations for \textsc{FacLoc} with Penalties in Figure \ref{fig:penalty-primal-dual}. For $j \in C$ and $i \in F$, define $w_{ij} \coloneqq \max\LR{\min \LR{r_i - c_{ij}, p_j - c_{ij}}, 0}$ and for $j \in C$, let $v_j \coloneqq \min_{i \in F} c_{ij} + w_{ij}$. Note that $v_j = \min_{i \in F}\max\LR{c_{ij}, \min\LR{r_i, p_j}}$. If $\i$ is a facility realizing the minimum $v_j = c_{\i j} + w_{\i j} = \max\LR{c_{\ij}, \min\LR{r_\i, p_j}}$, then we say that $\i$ is the bottleneck of $j$.
  
  To make the analysis easier, we consider a more expensive solution $(\tilde{C}', F')$ where the set of clients $\tilde{C}'$ is constructed using the following modified outlier determination phase: for each client $j$, if $\max\LR{r_\i, c_{\ij}} \le p_j$ then $j \in \tilde{C}'$ and otherwise $j \in \tilde{O}'$ where $\i$ is the bottleneck of $j$.
  
  It is easy to see that for any client $j \in C$, the ``cost'' paid by the client (i.e., connection cost, or its penalty) in the solution $(C', F')$ is at most the cost paid by it in the solution $(\tilde{C}', F')$. So henceforth, we analyze the cost of the solution $(\tilde{C}', F')$ by comparing it to the cost of a feasible dual \LP solution and in order to alleviate excessive notation, we will henceforth refer to the solution $(\tilde{C}', F')$ as $(C', F')$ and $\tilde{O}'$ as $O'$.
  
\begin{figure}
  \centering
  \makebox[\textwidth][c]{

    \begin{minipage}{0.65\textwidth}
      \begin{mdframed}[backgroundcolor=gray!9]
        \fontsize{9}{10}\selectfont
        \begin{alignat}{3}
          \text{minimize}   \displaystyle&\sum\limits_{i \in F} f_i y_i + \sum_{i \in F, j \in C} &c_{ij} x_{ij} + \sum_{j \in C} p_j z_j & \nonumber\\
          \text{subject to } \displaystyle& z_j + \sum\limits_{i \in F} x_{ij} \geq 1,   &\quad \forall j \in C \label{constr:penlpp-coverage} \\
          \displaystyle&\qquad\quad\ \   x_{ij} \le y_i, &\quad \forall  i \in F, \forall j \in C  \label{constr:penlpp-open}\\
          \displaystyle&\quad z_j, y_i, x_{ij} \in [0, 1] &\quad \forall i \in F,\forall j \in C \label{constr:penlpp-nonneg}
          \\\nonumber
        \end{alignat}
        \centering Primal LP
      \end{mdframed}
    \end{minipage}
    \makebox[0.65\textwidth][c]{
      \begin{minipage}{0.65\textwidth}
        \begin{mdframed}[backgroundcolor=gray!9]
          \fontsize{9}{9}\selectfont
          \begin{alignat}{3}
            \text{maximize}\displaystyle&\qquad \sum\limits_{j \in C} v_j
            \\\text{subject to } \displaystyle&\qquad\ \ v_j \le c_{ij} + w_{ij},   &\forall j \in C \label{constr:penlpd-alpha} \\
            \displaystyle&\ \sum_{j \in C} w_{ij} \le f_i &\forall i \in F \label{constr:penlpd-beta}\\
            \displaystyle&\ \  v_{j} \le p_j &\forall j \in C \label{constr:penlpd-penalty}\\
            \displaystyle& v_j, w_{ij} \ge 0 &\forall i \in F,\forall j \in C \label{constr:penlpd-integral}
          \end{alignat}
          \centering Dual LP
        \end{mdframed}
      \end{minipage}}
  }
  \caption{Primal and Dual Linear Programming Relaxations for \textsc{FacLoc} with Penalties \label{fig:penalty-primal-dual}}
\end{figure}
  
  Because of the way we choose the outliers in the solution we consider for the analysis \((C', F')\) we have the following property (where $\i$ is the bottleneck of $j$) --
  $$v_j = \begin{cases}
    \max\LR{c_{\ij}, r_\i} &\text{if } j \in C' \\
    p_j  &\text{if } j \in C \setminus C'
  \end{cases}$$
  
  We simultaneously prove feasibility of the dual solution we constructed, and show how it can be used to pay for the integral solution. We consider different cases regarding a fixed client $j \in C$ with bottleneck facility $\i$. We first prove a few straightforward claims.
  
  \begin{claim} \label{cl:p-1}
    If $\i \in F'$ is the bottleneck for $j \in C$, then $w_{i'j} = 0$ for all $i' \in F'$, where $i' \neq \i$.
  \end{claim}
  \begin{proof}
    Suppose there exists a facility $i' \in F'$ with $w_{i'j} > 0$. That is, $\min\LR{r_{i'} - c_{i'j}, p_j - c_{i'j}} > 0$, which further implies that $c_{i'j} < \min\LR{p_j, r_{i'}} \le r_{i'}$.
    
    If $c_{\ij} > r_{\i}$, then $v'_j = c_{\ij} + \max\LR{0, \min\LR{r_{\i} - c_{\ij}, p_j - c_{\ij} }} = c_{\ij}$. However, since $w_{i'j} > 0$, $c_{\ij} = v'_j \le c_{i'j} + w_{i'j} = \min\LR{p_j, r_{i'}} \le r_{i'}$
    
    Otherwise, $c_{\ij} \le r_{\i}$.
    
    Therefore in either case, $c_{\i i'} \le c_{\i j} + c_{j i'} \le 2 \max\LR{r_{i'}, r_{\i}}$, which is a contradiction since at most one of $\i, i'$ can belong to $F'$.
  \end{proof}
  
  \begin{claim} \label{cl:p-2}
    If $\i \not\in F'$ is the bottleneck of $j \in C$ and $\max\LR{r_\i, c_{\ij}} \le p_j$, and $i' \in F'$ caused $\i$ to close, then $c_{i' j} \le 3 v_j$.
  \end{claim}
  \begin{proof}
    Note that since we assume $\max\LR{r_\i, c_{\ij}} \le p_j$, we have $j \in C'$, $v_j = v'_j$, $c_{\ij} \ge p_j$, and $r_\i \le p_j$.
    Since $i' \in F'$ caused $\i$ to close, $r_{i'} \le r_{\i}$. Furthermore, $c_{i' \i} \le 2 r_{\i}$. Therefore, $c_{i' j} \le c_{\i i'} + c_{\i j} \le 2 r_{\i} + c_{\ij}$.
    
    If $c_{\ij} \ge r_{\i}$, then $w_{\ij} = 0$, and $c_{\ij} = v_j$. This means $2r_{\i} + c_{\ij} \le 3c_{\ij} = 3v_{j}$
    Otherwise, $c_{\ij} < r_{\i}$. Here, $v_j = \min\LR{r_{\i}, p_j}$ Then $2r_{\i} + c_{\ij} < 3r_{\i} = 3 \min\LR{r_{\i}, p_j} \le 3v_{j}$
    
    In either case, $c_{i' j} \le 3v_{j}$.
  \end{proof}

  \begin{claim} \label{cl:p-3}
    If $\i$ is the bottleneck of $j$, with $\max\LR{r_\i, c_{\ij}} \le p_j$, then $v_j \ge r_\i$.
  \end{claim}
  \begin{proof}
    Again, since we assume $\max\LR{r_\i, c_{\ij}} \le p_j$, we have $j \in C'$, $v_j = v'_j$, $c_{\ij} \ge p_j$, and $r_\i \le p_j$.
    
    Recall that $v_j = v'_j = \max\LR{c_{\ij}, \min\LR{r_\i, p_j}} = \max\LR{c_{\ij}, r_{\i}} \ge r_{\i}$ and the claim follows.
  \end{proof}
  
  We are now ready to prove the feasibility and approximation guarantee
  
  \begin{lemma}
    The solution $(v, w)$ is a feasible solution to the dual LP relaxation \ref{fig:penalty-primal-dual}.
  \end{lemma}
  \begin{proof}
    First note that constraints \ref{constr:penlpd-beta}, \ref{constr:penlpd-penalty}, and \ref{constr:penlpd-integral} are satisfied by construction for all $i \in F$ and $j \in C$ and so is constraint \ref{constr:penlpd-alpha} for all $j \in C'$.
    
    All that is left to show is that constraint \ref{constr:penlpd-alpha} is satisfied for all $j \in O'$. Since $j \in O'$, $\max\LR{r_\i, c_{\ij}} > p_j$.
    
    We have, $v_j \coloneqq p_j < \max\LR{r_\i, c_{\ij}} =  \max\LR{c_{\ij}, \min\LR{r_{\i}, p_j}} = v'_j \le c_{ij} + w_{ij}$ for any $i \in F$.
  \end{proof}

  \begin{lemma}
    For any $j \in C'$, there is some $i \in F'$ such that $c_{ij} + w_{ij} \le 3 v_j$.
  \end{lemma}
  \begin{proof}
    In all the cases, we assume that $\max\LR{r_{\i}, c_{\ij}} \le p_j$ and therefore $j \in C'$. This also implies $v_j = v'_j = \max \LR{c_{\ij}, \min\LR{p_j, r_{\i}}} = \max \LR{c_{\ij}, r_{\i}}$. Therefore, we can just disregard the penalties in the analysis.
    \begin{itemize}
      
    \item \textbf{Case 1.} $\i \in F'$
      
      Connect $j$ to $\i$. From claim \ref{cl:p-1}, we know that $w_{i'j} = 0$ for all other $i' \in F'$.
      
      \begin{enumerate}
      \item If $c_{\ij} < r_{\i}$, then $v_j = c_{\ij} + w_{\ij}$. In this case, $v_j$ pays for connecting $j$ to $\i$ and also for $j$'s contribution to opening cost of $\i$ which is exactly $w_{\ij}$.
        
      \item Otherwise $c_{\ij} \ge r_{\i}$, then $w_{\ij} = 0$, which is $j$'s contribution towards $\i$. We have $v_j = c_{\ij}$ and therefore $v_j$ pays for connecting $j$ to $\i$.
      \end{enumerate}
      
    \item \textbf{Case 2.} $\i \notin F'$ and $w_{ij} = 0$ for all $i \in F'$.
      
      Let $i'$ be the facility that caused $\i$ to close. Connect $j$ to $i'$. From claim \ref{cl:p-2}, we have $c_{i' j} \le 3v_{j}$. Therefore, $3v_j$ pays for the connection to $i'$.
      
    \item \textbf{Case 3.} $\i \notin F'$, there is some $i' \in F'$ with $w_{i'j} > 0$, but $i'$ did not cause $\i$ to close.
      
      We connect $j$ to $i'$. By assumption $w_{i'j} = r_{i'} - c_{i'j} > 0$. Furthermore, let $i$ be the facility that caused $\i$ to close. By claim \ref{cl:p-2} we have $c_{ij} \le 3 v_j$.
      
      We have $c_{i'j} + w_{i'j} = r_{i'}$. Also, $c_{i i'} > 2r_{i'}$, since $i', i$ both were added to $F'$.
      
      Now, $2(c_{i' j} + w_{i'j}) = 2r_{i'} \le c_{i i'} \le c_{i' j} + c_{i j}$. Subtracting $c_{i'j}$ from both sides, we get $c_{i' j} + 2w_{i'j} \le c_{i j} \le 3v_j$. Therefore, $3v_j$ pays for the connection cost of $j $ to $i'$ and also for (twice) $j$'s contribution towards opening $i'$.
      
    \item \textbf{Case 4.} $\i \notin F'$ and $i' \in F'$ with $w_{i'j} > 0$ caused $\i$ to close.
      
      We connect $j$ to $i'$. From claim \ref{cl:p-3}, we have that $v_j \ge r_\i$.
      
      Since $i'$ caused $\i$ to close, $r_\i \ge r_{i'} \ge c_{i'j} + w_{i'j}$. Therefore, $c_{i'j} + w_{i'j} \le r_{i'} \le r_{\i} \le v_j$. That is, $v_j$ pays for the connection cost of $j$ to $i'$, as well as its contribution towards opening of $i'$.
    \end{itemize}
  \end{proof}
  
  Thus, $(v, w)$ is a feasible dual solution. 
  We use the above analysis to conclude with the following theorem.
}

A primal-dual analysis of Algorithm \ref{alg:penaltiesSeq} leads to the following upper bound.

\begin{theorem}
$\cost(C', F') \le 3 \cdot \cost(C^*, F^*)$.
\end{theorem}
\ifthenelse{\boolean{short}}{}{
  \begin{proof}
    We show $\cost(C', F') \le 3 \lr{\sum_{j \in C} v_j}$, which is sufficient since $(v, w)$ is a feasible dual solution, and cost of any feasible dual solution is a lower bound on the cost of an integral optimal solution.
    
    As we have argued previously, for any $j \in C \setminus C'$, we have $p_j = v_j$, and that for any $j \in C'$, we have $d(j, F') + s(j)$, where $s(j) \ge 0$ is the contribution of $j$ towards opening a single facility in $F'$. We have also argued that any $j \in C'$ contributes $s(j)$ for at most one open facility from $F'$. It follows that,
    
    \begingroup
    \allowdisplaybreaks
    \begin{align*}
      \cost(C', F') &= \sum_{i \in F'} f_i + \sum_{j \in C'} d(j, F') + \sum_{j \in C \setminus C'} p_j
      \\&= \lr{\sum_{j \in C'} s(j) + d(j, F')} + \sum_{j \in C \setminus C'} p_j
      \\&\le 3 \sum_{j \in C'} v_j + \sum_{j \in C \setminus C'} v_j
      \\&\le 3 \lr{\sum_{j \in C} v_j}
    \end{align*}
    \endgroup
  \end{proof}
}
 \section{Distributed Robust Facility Location: Implicit Metric}

We first present our $k$-machine algorithm for Distributed Robust \textsc{FacLoc} in the implicit metric setting and derive the Congested Clique as a special case for $k = n$. We then describe how to implement the algorithm in the MPC model.

\subsection{The \texorpdfstring{$k$}{k}-Machine Algorithm}

In this section we show how to implement the sequential algorithms for the Robust \textsc{FacLoc} in the \kmm. To do this we first need to establish some primitives and techniques. These have largely appeared in \cite{BandyapadhyayArxiv2017}. Then we will provide details for implementing the Robust \textsc{FacLoc} algorithm in the \kmm.

Since the input metric is only implicitly provided, as an edge-weighted graph, a key primitive that we require is computing shortest path distances to learn parts of the metric space. To this end, the following lemma shows that we can solve the Single Source Shortest Paths (SSSP) problem efficiently in the \kmm.

\begin{lemma}[Corollary 1 in \cite{BandyapadhyayArxiv2017}]
\label{lem:SSSP}
For any $0 < \epsilon \le  1$, there is a deterministic $(1 + \epsilon)$-approximation algorithm in the $k$-machine model for solving
the SSSP problem in undirected graphs with non-negative edge-weights in
	$O((n/k) \cdot \mbox{poly}(\log n)/\mbox{poly}(\epsilon))$ rounds.
\end{lemma}

In addition to SSSP, our algorithms require an efficient solution to a more general problem that we call \textit{Multi-Source Shortest Paths} (in short, MSSP) \ifthenelse{\boolean{short}}{}{and a variant of MSSP that we call \textsc{ExclusiveMSSP}}. The input is an edge-weighted graph $G = (V, E)$, with non-negative edge-weights, and a set $T \subseteq V$ of sources.

For MSSP, the output is required to be, for each vertex $v$, the distance $d(v, T)$ (i.e., $\min\{d(v, u) \mid u \in T\}$) and the vertex $v^* \in T$ that realizes this distance. \ifthenelse{\boolean{short}}{The following lemma shows that we can solve this problem efficiently in the \kmm.}{Whereas in \textsc{ExclusiveMSSP}, for each $v \in T$, we are required to output $d(v, T\setminus\{v\})$ and the vertex $u^* \in T \setminus \{v\}$ that realizes this distance. The following two lemmas show that we can solve these two problems efficiently in the \kmm.}

\begin{lemma}[Lemma 4 in \cite{BandyapadhyayArxiv2017}]
  \label{lem:MSSP}
  Given a set \(T \subseteq V\) of sources known to the machines (i.e., each machine $m_j$ knows $T \cap H(m_j)$), we can, for any value $0 \le \epsilon \le 1$,
  compute a $(1 + \epsilon)$-approximation to MSSP in $\tilde{O}(1/\mbox{poly}(\epsilon) \cdot n/k)$ rounds, w.h.p.
  Specifically, after the algorithm has ended, for each $v \in V \setminus T$, the machine $m_j$ that hosts $v$ knows a pair
  $(u, \tilde{d}) \in T \times \mathbb{R}^+$, such that
  \(d(v, u) \le \tilde{d} \le (1+\epsilon) \cdot d(v, T)\).
\end{lemma}

\ifthenelse{\boolean{short}}{}{
\begin{lemma}[Lemma 5 in \cite{BandyapadhyayArxiv2017}]
  \label{lem:ExclusiveMSSP}
  Given a set \(T \subseteq V\) of sources known to the machines (i.e., each machine $m_j$ knows $T \cap H(m_j)$), we can, for any value $0 \le \epsilon \le 1$,
  compute a $(1 + \epsilon)$-approximation to \textsc{ExclusiveMSSP} in $\tilde{O}(1/\mbox{poly}\,(\epsilon) \cdot n/k)$ rounds, w.h.p.
  Specifically, after the algorithm has ended, for each $v \in T$, the machine $m_j$ that hosts $v$ knows a pair
  $(u, \tilde{d}) \in T \setminus \{v\} \times \mathbb{R}^+$, such that $d(v, u) \le \tilde{d} \le (1+\epsilon) \cdot d(v, T \setminus \{v\})$.
\end{lemma}}

\ifthenelse{\boolean{short}}{
  Using the primitives described above, \cite{BandyapadhyayArxiv2017} show that it is possible to compute approximate radius values efficiently in the \kmm. The algorithm is described here, see the full version \cite{InamdarPPArxiv2018} for details on the implementation of this algorithm.

For any facility or client $v$ and for any integer $i \ge 1$, let $q_i(v)$ denote $|B(v, (1+\epsilon)^i)|$, the size of the neighborhood of 
$v$ within distance $(1+\epsilon)^i$.

\RestyleAlgo{boxruled}  
\begin{algorithm2e}\caption{\textsc{RadiusComputation} Algorithm\label{alg:RC}}
	\textbf{Neighborhood-Size Computation.} Each machine $m_j$ computes $q_i(v)$, for all integers $i \ge 0$ and for all vertices $v \in H(m_j)$.\\
	\textbf{Local Computation.} Each machine $m_j$ computes $\tilde{r}_v$ locally, for all vertices $v \in H(m_j)$. (Recall that $\tilde{r}_v := (1+\epsilon)^{t-1}$ where $t \ge 1$ is the smallest integer for which $\sum_{i = 0}^t q_i(v) \cdot ((1+\epsilon)^{i+1} - (1+\epsilon)^i) > f_v$.)
\end{algorithm2e}
   
Therefore, we get the following lemma the proof of which can be found in Section 4 of \cite{BandyapadhyayArxiv2017} --

\begin{lemma} \label{lem:apxrvalues}
For each facility \(v \in F\) it is possible to compute an approximate radius \(\tilde{r}_v\) in \(\tilde{O}(n/k)\) rounds of the \kmm such that $\frac{r_v}{(1+\epsilon)^2} \le \tilde{r}_v \le (1+\epsilon)^2 r_v$ where \(r_v\) is the actual radius of \(v\) satisfying $f_v = \sum_{u \in B(v, r_v)} (r_v - d(v, u))$.
\end{lemma}

The greedy phase is implemented by discretizing the radius values computed in the first phase which results in \(O(\log_{1 + \epsilon} n)\) distinct categories. Note that in each category, the order in which we process the facilities does not matter as it will only add an extra \((1 + \epsilon)\) factor to the approximation ratio. This reduces the greedy phase to computing a \textit{maximal independent set (MIS)} on a suitable intersection graph for each category \(i\) where the vertices are the facilities in the \(i^{th}\) category and there is an edge between two vertices if they are within distance \(2(1+\epsilon)^i\) of each other.

Finding such an MIS requires \(O(\log n)\) calls to a subroutine that solves MSSP \cite{Thorup2001} and since our implementation of MSSP only returns approximate distances, what we really compute is a relaxed version of an MIS called an \((\epsilon, d)\)-MIS in \cite{BandyapadhyayArxiv2017}.

\begin{definition}[$(\epsilon, d)$-approximate MIS]
  For an edge-weighted graph $G = (V, E)$, and parameters $d, \epsilon > 0$, an $(\epsilon, d)$-approximate MIS is a subset $I \subseteq V$ such that
  \begin{enumerate}
  \item For all distinct vertices $u, v \in I$, $d(u, v) \ge \frac{d}{1+\epsilon}$.
  \item For any $u \in V \setminus I$, there exists a $v \in I$ such that $d(u, v) \le d \cdot (1+\epsilon)$.
  \end{enumerate}
\end{definition}

The work in \cite{BandyapadhyayArxiv2017} gives an algorithm that efficiently computes an approximate MIS of an induced subgraph \(G[W]\) of \(G\) for any vertex set \(W\) in the \kmm.

\begin{lemma} \label{lem:DistdMIS}
  We can find an $(O(\epsilon), d)$-approximate MIS $I$ of $G[W]$ whp in \(\tilde{O}(n/k)\) rounds.
\end{lemma}
 }
 {
\subsubsection{Radius Computation}
Using the primitives we described in the previous section, \cite{BandyapadhyayArxiv2017} show that it is possible to compute approximate radius values efficiently in the \kmm by computing neighborhood size estimates along the lines of \cite{Cohen1997, Thorup2001}. A version of the algorithm is described in \ref{alg:RC}. We discuss the implementation of this algorithm in a fair bit of detail because we will need to modify certain aspects when implementing the \textsc{FacLoc} with Penalties algorithm in the \kmm (Section \ref{sec:pen-imp-radcomp}).

For any facility or client $v$ and for any integer $i \ge 1$, let $q_i(v)$ denote $|B(v, (1+\epsilon)^i)|$, the size of the neighborhood of 
$v$ within distance $(1+\epsilon)^i$.
\RestyleAlgo{boxruled}  
\begin{algorithm2e}\caption{\textsc{RadiusComputation} Algorithm\label{alg:RC}}
	\textbf{Neighborhood-Size Computation.} Each machine $m_j$ computes $q_i(v)$, for all integers $i \ge 0$ and for all vertices $v \in H(m_j)$.\\
	\textbf{Local Computation.} Each machine $m_j$ computes $\tilde{r}_v$ locally, for all vertices $v \in H(m_j)$. (Recall that $\tilde{r}_v := (1+\epsilon)^{t-1}$ where $t \ge 1$ is the smallest integer for which $\sum_{i = 0}^t q_i(v) \cdot ((1+\epsilon)^{i+1} - (1+\epsilon)^i) > f_v$.)
\end{algorithm2e}

In Algorithm \ref{alg:RC}, step 2 is just local computation, so we focus on Step 1 which requires the solution to the problem of computing neighborhood sizes.

Cohen's algorithm starts by assigning to each vertex $v$ a rank \(\r(v)\) chosen uniformly from \([0, 1]\). These ranks induce a random permutation of the vertices. To compute the size estimate of a neighborhood, say $B(v, d)$, for a vertex $v$ and real $d > 0$, Cohen's algorithm finds the smallest rank of a vertex in $B(v, d)$. It is then shown (in Section 6, \cite{Cohen1997}) that the expected value of the smallest rank in $B(v, d)$ is $1/(1 + |B(v, d)|)$. Thus, in expectation, the reciprocal of the smallest rank in $B(v, d)$ is (almost) identical to $|B(v, d)|$. To obtain a good estimate of $|B(v, d)|$ with high probability, Cohen simply repeats the above-described procedure independently a bunch of times and shows the following concentration result on the average estimator.

\begin{theorem} \textbf{(Cohen \cite{Cohen1997})} \label{thm:cohen}
	Let $v$ be a vertex and $d > 0$ a real.
	For $1 \le i \le \ell$, let $R_i$ denote the smallest rank of a vertex in $B(v, d)$ obtained in the $i^{th}$ repetition of
	Cohen's neighborhood-size estimation procedure. 
	Let $\hat{R}$ be the average of $R_1, R_2, \ldots, R_\ell$.
	Let $\mu = 1/(1 + |B(v, d)|)$. Then, for any $0 < \epsilon < 1$,
	$$\Pr(|\hat{R} - \mu| \ge \epsilon \mu) = \exp(-\Omega(\epsilon^2 \cdot \ell)).$$
\end{theorem}
This theorem implies that $\ell = O(\log n/\epsilon^2)$ repetitions suffice for obtaining $(1 \pm \epsilon)$-factor estimates w.h.p.~of the sizes of $B(v, d)$ for all $v$ and all $d$. 

In \cite{BandyapadhyayArxiv2017}, the authors show that Algorithm \ref{alg:CohenEstimates} can simulate Cohen's neighborhood size estimation framework in the \kmm in \(\tilde{O}(n/k)\) rounds

\RestyleAlgo{boxruled}
\begin{algorithm2e}\caption{\textsc{NbdSizeEstimates}\((G, \epsilon)\)\label{alg:CohenEstimates}}
	$\epsilon' := \epsilon/(\epsilon + 4)$; $t = \lceil 2 \log_{1+\epsilon'} n \rceil$; $\ell := \lceil c \log n/(\epsilon')^2 \rceil$\\
  \For{\(j := 1, \ldots, \ell\)} {
	  \textbf{Local Computation.} Each machine $m_j$ picks a rank $\r(v)$, for each vertex $v \in H(m_j)$,
	  chosen uniformly at random from $[0, 1]$. Machine $m_j$ then rounds $\r(v)$ down to the closest 
	  $(1+\epsilon')^i/n^2$ for integer $i \ge 0$\label{alg2:ChooseRanks} \\
    \For{$i := 0, 1, \ldots, t-1$}{
      $T_i := \{v \in V \mid \r(v) = (1+\epsilon')^i/n^2\}$\\
	Compute a \((1+\epsilon)\)-approximate solution to MSSP using \(T_i\) as the set of sources \label{alg2:MSSP}; let $\tilde{d}(v, T_i)$ denote the computed approximate distances\\
	\textbf{Local Computation.} Machine $m_j$ stores $\tilde{d}(v, T_i)$ for each \(v \in H(m_j)\)\\
      }
  }
\end{algorithm2e}

Therefore, we get the following lemma the proof of which can be found in Section 4 of \cite{BandyapadhyayArxiv2017}.

\begin{lemma} \label{lem:apxrvalues}
For each facility \(v \in F\) it is possible to compute an approximate radius \(\tilde{r}_v\) in \(\tilde{O}(n/k)\) rounds of the \kmm such that $\frac{r_v}{(1+\epsilon)^2} \le \tilde{r}_v \le (1+\epsilon)^2 r_v$ where \(r_v\) is the actual radius of \(v\) satisfying $f_v = \sum_{u \in B(v, r_v)} (r_v - d(v, u))$.
\end{lemma}

\subsubsection{Greedy Phase}
The greedy phase is implemented by discretizing the radius values computed in the first phase which results in \(O(\log_{1 + \epsilon} n)\) distinct categories. Note that in each category, the order in which we process the facilities does not matter as it will only add an extra \((1 + \epsilon)\) factor to the approximation ratio. This reduces the greedy phase to computing a \textit{maximal independent set (MIS)} on a suitable intersection graph for each category \(i\) where the vertices are the facilities in the \(i^{th}\) category and the there is an edge between two vertices if they are within distance \(2(1+\epsilon)^i\) of each other.

Finding such an MIS requires \(O(\log n)\) calls to a subroutine that solves MSSP \cite{Thorup2001} and since our implementation of MSSP only returns approximate distances, what we really compute is a relaxed version of an MIS called an \((\epsilon, d)\)-MIS in \cite{BandyapadhyayArxiv2017}.

\RestyleAlgo{boxruled}
\begin{algorithm2e}\caption{\textsc{ApproximateMIS}\((G, W, d, \epsilon)\)\label{alg:DistdMIS}}
  Each machine $m_j$ initializes \(U_j := \emptyset\) \\
  \tcc{Let $W_j$ denote $W \cap H(m_j)$.}
  \For{$i := 0, 1, \ldots, \lceil \log n \rceil$}{
    \For{$\lceil c \log n \rceil$ iterations}{
      Each machine \(m_j\) marks each vertex $v \in W_j$ with probability $2^i/n$ \\
      \tcc{Let \(R_j \subset W_j\) denote the set of marked vertices hosted by $m_j$, let $R := \cup_{j=1}^k R_j$} 
      Solve an instance of the \textsc{ExclusiveMSSP} problem using $R$ as the set of sources (see Lemma \ref{lem:ExclusiveMSSP}) to obtain $(1+\epsilon)$-approximate distances $\tilde{d}$ \label{alg1:MSSP1}\\
      Each machine $m_j$ computes $T_j := \{v \in R_j \mid \tilde{d}(v, R \setminus \{v\}) > d\}$ \\
      Each Machine $m_j$ sets $U_j := U_j \cup T_j$ \\
      \tcc{Let $T := \cup_{j=1}^k T_j$}
      Solve an instance of the \textsc{MSSP} problem using $T$ as the set of sources (see Lemma \ref{lem:MSSP}) to obtain $(1+\epsilon)$-approximate distances $\tilde{d}$ \label{alg1:MSSP2} \\
      Each machine $m_j$ computes $Q_j = \{v \in W_j \mid \tilde{d}(v, T) \le (1 + \epsilon)d\}$ \\
      Each machine $m_j$ sets $W_j := W_j \setminus (T_j \cup Q_j)$
    }
  }
  \textbf{return} \(U := \cup_{j = 1}^k U_j\)
\end{algorithm2e}

\begin{definition}[$(\epsilon, d)$-approximate MIS]
  For an edge-weighted graph $G = (V, E)$, and parameters $d, \epsilon > 0$, an $(\epsilon, d)$-approximate MIS is a subset $I \subseteq V$ such that
  \begin{enumerate}
  \item For all distinct vertices $u, v \in I$, $d(u, v) \ge \frac{d}{1+\epsilon}$.
  \item For any $u \in V \setminus I$, there exists a $v \in I$ such that $d(u, v) \le d \cdot (1+\epsilon)$.
  \end{enumerate}
\end{definition}

The work in \cite{BandyapadhyayArxiv2017} gives an algorithm that efficiently computes an approximate MIS in the \kmm which we describe in Algorithm \ref{alg:DistdMIS}.

\begin{lemma} \label{lem:DistdMIS}
  Algorithm \ref{alg:DistdMIS} finds an $(O(\epsilon), d)$-approximate MIS $I$ of $G[W]$ whp in \(\tilde{O}(n/k)\) rounds.
\end{lemma}
}

We are now ready to describe the \kmm implementation of Algorithm \ref{alg:outlierSeq}.

\RestyleAlgo{boxruled}
\begin{algorithm2e}\caption{\textsc{RobustFacLocDist}\((F, C, p)\)\label{alg:outlierskmm}}
  \tcc{Recall $\ell \coloneqq |C| - p$}
  \For{$t = 1, \ldots, O(\log n)$}{
    Let $i_e \in F$ be a most expensive facility from the facilities with opening costs in the range $\left[(1+\epsilon)^t, (1+\epsilon)^{t+1}\right)$ \label{lin:i_e-kmm} \\
    Modify the facility opening costs to be
    $$f'_i = \begin{cases}
      +\infty &\text{ if $f_i > f_{i_e}$}
      \\0 &\text{ if $i = i_e$}
      \\f_i &\text{ otherwise}
    \end{cases}$$ \\
    \tcc{Radius Computation Phase:}
    Call the \textsc{RadiusComputation} algorithm (Algorithm \ref{alg:RC}) to compute  approximate radii. \label{alg4:radius-computation}\\
    
    \tcc{Greedy Phase:}
    Let \(F' = \emptyset\),  \(C' = \emptyset\), \(O' = C\) \\
    \For{\(i = 0, 1, 2, \dots\)}{ \label{alg4:for}
      Let \(W\) be the set of vertices \(w \in F\) across all machines with \(\tilde{r}_w = \tilde{r} = (1 + \epsilon)^i\) \\
      Using Lemma \ref{lem:MSSP}, remove all vertices from \(W\) within approximate distance \(2(1 + \epsilon)^3\cdot \tilde{r}\) from \(F'\) \label{alg4:remove-close-facs}\\
      \(I \leftarrow \textsc{ApproximateMIS}(G, W, 2(1 + \epsilon)^3 \cdot \tilde{r}, \epsilon)\) \\
      \(F' \leftarrow F' \cup I\) \\
      Using Lemma \ref{lem:MSSP}, move from \(O'\) to \(C'\) all vertices that are within distance \((1 + \epsilon) \cdot \tilde{r}\) from \(F_i\), the set of facilities processed up to iteration \(i\) \label{alg4:remove-clients}\\
      \textbf{if} $|O'| \le \ell$ \textbf{then break} \label{alg4:break} \\
    }\label{alg4:endfor}
    
    \tcc{Outlier Determination Phase:}
    \If{$|O'| > \ell$} {
      Using Lemma \ref{lem:MSSP} find $O_1 \subseteq O'$, a set of $|O'| - \ell$ clients that are closest to facilities in $F'$. \\
      $C' \gets C' \cup O_1$, \quad $O' \gets O' \setminus O_1$.
    }
    \ElseIf{$|O'| < \ell$} {
      Using Lemma \ref{lem:MSSP} find $O_2 \subseteq C \setminus O'$, a set of $(\ell - |O'|)$ clients that are farthest away from facilities in $F'$ \\
      $C' \gets C' \setminus O_2$, \quad $O' \gets O' \cup O_2$.
    }
    Let $(C'_t, F'_t) \gets (C', F')$
  }
  \Return $(C'_t, F'_t)$ with a minimum cost \label{lin:minimum-cost}
\end{algorithm2e}

Our \kmm implementation of the Robust \textsc{FacLoc} algorithm is summarized in Algorithm \ref{alg:outlierskmm}. The correctness proof is similar to that of Algorithm \ref{alg:outlierSeq} but is complicated by the fact that we compute \((1 + \epsilon)\)-approximate distances instead of exact distances. Again, as in the analysis of the sequential algorithm, we abuse the notation so that (i) $(C', F')$ refers to a minimum-cost solution returned by the algorithm, (ii) $i_e$ refers to the facility chosen in the line \ref{lin:i_e-kmm} of the algorithm, and (iii) the modified instance with original facility costs. This analysis appears in the \ifthenelse{\boolean{short}}{full version \cite{InamdarPPArxiv2018}}{next section}, and as a result we get the following theorem.

\begin{theorem} \label{thm:kmm-OutlierGuarantee}
  In \(\tilde{O}(\text{poly}(1/\epsilon) \cdot n/k)\) rounds, whp, Algorithm \ref{alg:outlierskmm} finds a factor \(5 + O(\epsilon)\) approximate solution \((C', F')\) to the Robust \textsc{FacLoc} problem for any constant \(\epsilon > 0\).
\end{theorem}
\ifthenelse{\boolean{short}}{}{
  \begin{proof}
    There are \(O(\log_{(1 + \epsilon)} \frac{f_{\max}}{f_{\min}}) = O(\log n)\) iterations of the outer for loop, where a facility with the highest opening cost from the range $\left[ (1+\epsilon)^t, (1+\epsilon)^{t+1} \right)$. The guess can be broadcast to all the machines, and they can modify their part of the instance appropriately (without actually removing the facilities from the metric graph). This extra factor is absorbed by the tilde notation, provided that each iteration of the for loop takes $\tilde{O}(n/k)$ rounds. We can also estimate the cost of a solution within a factor of $(1+O(\epsilon))$ factor in $\tilde{O}(n/k)$ rounds -- the details can be found in \cite{BandyapadhyayArxiv2017}. Since there are $O(\log n)$ candidate solutions to find a minimum-cost solution from, in line \ref{lin:minimum-cost}, this step can also be implemented in $\tilde{O}(n/k)$ rounds. 
    
    Each iteration of the for loop \ref{alg:outlierskmm} consists of two phases namely, the Radius Computation and Greedy Phases. We bound the running time of both these phases separately. By Lemma \ref{lem:apxrvalues} we know that the radius computation phase of Algorithm \ref{alg:outlierskmm} requires \(\tilde{O}(n/k)\) rounds. In the for loop on line \ref{alg4:for} there are at most $O(\log_{1 + \epsilon}{nN}) = O(\log nN) = O(\log n)$ possible values of \(i\) and hence at most $O(\log n)$ iterations (where \(N = \mbox{poly}(n)\) is the largest edge weight). Each individual step in the greedy phase of Algorithm \ref{alg:outlierskmm} takes \(\tilde{O}(n/k)\) rounds therefore we conclude that the overall running time is \(\tilde{O}(n/k)\) rounds.The proof of the approximation guarantee appears in the \ifthenelse{\boolean{short}}{full version \cite{InamdarPPArxiv2018}}{next section}.
  \end{proof}

  \subsubsection{Analysis of the Algorithm}

  Similar to the sequential algorithm analysis, we analyze the cost of the corresponding costlier solution $(\tilde{C}', F')$. In order to alleviate excessive notation, we will henceforth refer to the solution $(\tilde{C}', F')$ as $(C', F')$ and $\tilde{O}'$ as $O'$. We now restate the standard primal and dual for the Robust Facility Location problem.

  Let \(r_i\) be the radius value of \(i\) satisfying $f_i = \sum_{j \in B(i, r_i)} (r_i - c_{i j})$ and let \(\tilde{r}_i\) be the approximate radius value of \(i\) computed during Algorithm \ref{alg:outlierskmm}
  
  First, we construct a feasible dual solution $(v, w, \q)$. For a facility $i \in F$ and client $j \in C$, let  $w_{ij} \coloneqq \max\{0, r_i - c_{ij}\}$. Let $\q \coloneqq \max_{j \in C'} v'_{j}/(1 + \epsilon)^4$ (recall that  $v'_j \coloneqq \min_{i \in F} \max\{r_i, c_{ij}\} = \min_{i \in F} c_{ij} + w_{ij}$). Now, for a client $j \in C$, define $v_j$ as follows:
  $$v_j = \begin{cases}
    v'_j/(1 + \epsilon)^4 &\text{if } j \in C'
    \\\q &\text{if } j \in O'
  \end{cases}$$

  \begin{claim} \label{cl:apdx-bound_vj}
    If a client \(j \in C''\) then \(v_j' \le (1 + \epsilon)^2 \tilde{r}_i\) and if a client \(j \in O''\) then \(v_j' \ge (1 + \epsilon)^{-2} \tilde{r}_i\) where \(i\) is the last iteration of the for loop (lines \ref{alg4:for}-\ref{alg4:endfor})
  \end{claim}
  \begin{proof}
    If \(j \in C''\) then it must be added to \(C_{i'}\) for some iteration \(i' \le i\). Let us assume wlog that \(j\) was added to \(C_i\). Therefore, there must be some \(i' \in F_i\) such that \(j \in B(i', (1 + \epsilon) \tilde{r}_i)\). This means that \(c_{i'j} \le  (1 + \epsilon) \tilde{r}_i\) 
    \begin{align*}
      v'_j &= \min_{i \in F}{\max\{r_{i}, c_{ij}\}} \le \max\{r_{i'}, c_{i'j}\} \\
           &\le \max\{(1 + \epsilon)^{2} \tilde{r}_{i'}, (1 + \epsilon) \tilde{r}_{i'}\} \\
           &\le (1 + \epsilon)^2 \tilde{r}_i \tag{since we process facilities in increasing value of \(\tilde{r}\)}
    \end{align*}
    
    Since we compute approximate shortest paths, if client \(j\) is not added to any \(C_i\) then for all \(i' \in F_i\), \(j \notin B(i', \tilde{r}_i)\) (otherwise we would add \(j\) to \(C_i\)). Therefore if \(j \in O''\), for all \(i' \in F_i\), \(c_{i'j} > \tilde{r}_i\). So we have,
    \[v'_j = \min_{i \in F}{\max\{r_{i}, c_{ij}\}} \ge (1 + \epsilon)^{-2} \min_{i' \in F_i}{\max\{r_{i'}, c_{i'j}\}}\]
    
    Because for facilities \(i'' \in F \setminus F_i\), \(r_{i''} \ge (1 + \epsilon)^{-2} \tilde{r}_{i''} \ge (1 + \epsilon)^{-2} \tilde{r}_i\) as we process facilities in increasing order of \(\tilde{r}\). Therefore,
    \begin{align*}
      v'_j &\ge (1 + \epsilon)^{-2} \min_{i' \in F_i}{\max\{r_{i'}, \tilde{r}_i\}} \\
           &\ge (1 + \epsilon)^{-2} \min_{i' \in F_i}{\max\{(1 + \epsilon)^{-2}\tilde{r}_{i'}, \tilde{r}_i\}} \\
           &= (1 + \epsilon)^{-2} \tilde{r}_{i}
    \end{align*}
  \end{proof}

  \begin{lemma}
    The solution $(v, w, \q)$ is a feasible solution to the dual LP relaxation \ref{fig:outlier-primal-dual}.
  \end{lemma}
  \begin{proof}
    Note that constraints \ref{constr:fllpd-beta}, \ref{constr:fllpd-alphaub}, and \ref{constr:fllp-integral} of the dual are satisfied by construction and so is constraint \ref{constr:flld-alpha} for clients $j \in C'$. Therefore, in order to show that the solution $(v, w, \q)$ is feasible, we have to show that constraint \ref{constr:flld-alpha} is satisfied for all clients $j \in O'$. To this end, we consider the following two cases.
    
    \textbf{Case 1.} We enter the outlier determination phase after iterating over all facilities in $F$. Therefore, we have $|O''| > \ell$. This means that we identified a set \(O_1 \subset O''\) of size \(|O''| - \ell\) to be marked as non-outliers.
    
    As we iterate over all the facilities, by Claim \ref{cl:apdx-bound_vj}, for \(j \in O''\), we get $v'_j \ge \max_{i \in F}{\tilde{r}_i/(1 + \epsilon)^2}$.
    
    We put in \(O_1\) the clients from \(O''\) with smallest \(v'_j\)-value \footnote{Note this is not what the algorithm does but we have argued that the algorithm's solution is better than the solution we are analyzing}. This means that for all \(j \in O'\):
    \begin{align*}
      v'_j &\ge \max_{j \in O_1}{v'_j} \\
           &= \max\LR{\max_{j \in C''}{\frac{v'_j}{(1 + \epsilon)^4}}, \max_{j \in O_1}{v'_j}} \tag{\(\max_{j \in C''} {v'_j} \le \max_{i \in F} {(1 + \epsilon)^2\tilde{r}_i}\) by Claim \ref{cl:apdx-bound_vj}} \\
           &\ge \frac{1}{(1 + \epsilon)^4} \max\LR{\max_{j \in C''}{v'_j}, \max_{j \in O_1}{v'_j}} \\
           &\ge \frac{1}{(1 + \epsilon)^4} \max_{j \in C'}{v'_j} \\
           &\ge \q
    \end{align*}
    
    Therefore, we conclude that for any $j \in O'$ and $i \in F$, $c_{ij} + w_{ij} \ge v'_j \ge \q = v_j$.
    
    \textbf{Case 2.} We enter the outlier determination phase because of the break statement on line \ref{alg4:break}. Here, $|O''| \le \ell$ and \(C' \subseteq C''\)
    
    Let \(i^*\) be the last iteration of the for loop. Therefore \(F_{i^*}\) is the set of facilities we consider in the for loop and \(\max_{i \in F_{i^*}} {r_i} \le (1 + \epsilon)^2\tilde{r}_{i^*}\) We show that $\q \le (1 + \epsilon)^2\tilde{r}_{i^*}$.

    Recall that by the case assumption we have $|O''| \le \ell$ and hence \(C' \subseteq C''\). All clients $j \in C''$ were part of \(C_i\) for some \(i \in F_{i^*}\) and by Claim \ref{cl:apdx-bound_vj} we have \(v'_j \le (1 + \epsilon)^2 \tilde{r}_i \le (1 + \epsilon)^2 \tilde{r}_{i^*}\). Therefore,
    $$\q = \max_{j \in C'} \frac{v'_j}{(1 + \epsilon)^4} \le \max_{j \in C''} \frac{v'_j}{(1 + \epsilon)^4} \le \max_{i \in F_{i^*}} \frac{\tilde{r}_i}{(1 + \epsilon)^2} = \frac{\tilde{r}_{i^*}}{(1 + \epsilon)^2}$$
    
    Let $j \in O'$ be an outlier client. If $j \in O''$, then for any facility $i \in F$,
    \begin{align*}
      c_{ij} + w_{ij} &\ge v'_j \\
                      &\ge (1 + \epsilon)^{-2} \tilde{r}_{i^*} \tag{by Claim \ref{cl:apdx-bound_vj}} \\
                      &\ge \q \\
                      &= v_j
    \end{align*}
    
    Otherwise, $j \in O_2$ and was added to $O'$ because it had highest $v'_j$ value among $C_{i^*}$. Therefore, by Claim \ref{cl:apdx-bound_vj} it follows that for any facility $i \in F$, $c_{ij} + w_{ij} \ge v'_j \ge \max_{j \in C'} v'_{j} (1 + \epsilon)^{-4} = \q = v_j$
    
    From the case analysis it follows that $v_j \le c_{ij} + w_{ij}$ for all $j \in O'$ and for all $i \in F$. Therefore, we have shown that $(v, w, \q)$ is a dual feasible solution.
  \end{proof}
  
  For the approximation guarantee we can now focus just on the clients in $C'$ because the only contribution that the clients in $O'$ make to the dual objective function is to cancel out the $- \ell \q$ term and hence they do not affect the approximation guarantee.We call a facility $\i$ the \emph{bottleneck} of $j$ if \(v'_j = \max\{r_{\i}, c_{\i j}\} =  c_{\i j} + w_{\i j}\).
  
  Throughout this section, we condition on the event that the outcome of all the randomized algorithms is as expected (i.e. the ``bad'' events do not happen). Note that this happens with w.h.p. We first need the following facts along the lines of \cite{Thorup2001}. 
  
  \begin{lemma}[Modified From Lemma 8 Of \cite{Thorup2001}] \label{lem:apdx-lem8}
    There exists a total ordering $\prec$ on the facilities in $F$ such that $u \prec v \implies \tilde{r}_u \le \tilde{r}_v$, and $v$ is added to $F'$ if and only if there is no previous $u \prec v$ in $F'$ such that $c_{u v} \le 2(1+\epsilon)^2 \tilde{r}_v$.
  \end{lemma}
  \begin{proof}[Proof Sketch]
    The ordering is obtained by enumerating the facilities processed in each iteration (with arbitrary order given to facilities in the same iteration). The facilities in $I$ that are included in $F'$ before the rest of the vertices of $W$. The lemma follows because of line \ref{alg4:remove-close-facs} (if $u$ and $v$ are processed in different iterations) the definition of $(\epsilon, d)$-approximate MIS (if $u$ and $v$ are processed in the same iteration).
  \end{proof}

  \begin{claim}[Modified From Claim 9.2 Of \cite{Thorup2001}] \label{cl:apdx-cl9.2}
    For any two distinct vertices $u, v \in F'$, we have that $c_{u v} > 2(1+\epsilon)^2 \cdot \max\{\tilde{r}_u, \tilde{r}_v\}$.
  \end{claim}
  \begin{proof}[Proof Sketch]
    Without loss of generality, assume that $u \prec v$, so $\tilde{r}_u \le \tilde{r}_v$. From Lemma \ref{lem:apdx-lem8} we have $c_{u v} > 2(1+\epsilon)^2 \tilde{r}_v \ge 2(1+\epsilon)^2 \tilde{r}_v$
  \end{proof}

  We now prove a few claims about the dual solution.
  
  \begin{claim} \label{cl:apdx-apx-1}
    For any $i \in F$ and $j \in C$, $\tilde{r}_i \le (1 + \epsilon)^2 (c_{ij} + w_{ij})$. Furthermore if for some \(i \in F\) and \(j \in C\), \(w_{ij} > 0\), then \(\tilde{r}_i \ge (1 + \epsilon)^{-2} (c_{ij} + w_{ij})\)
  \end{claim}
  \begin{proof}
    We have $w_{ij} = \max\{0, r_i - c_{ij}\} \ge r_i - c_{ij}$ and therefore $\tilde{r}_i \le (1 + \epsilon)^2 r_i \le (1 + \epsilon)^2 (c_{ij} + w_{ij})$.
    
    If for some \(i \in F\) and \(j \in C\), \(w_{ij} > 0\), then \(w_{ij} = r_i - c_{ij}\) which means that \(\tilde{r}_i \ge (1 + \epsilon)^{-2} r_i = (1 + \epsilon)^{-2} (c_{ij} + w_{ij})\)
  \end{proof}
  
  \begin{claim} \label{cl:apdx-apx-2}
    If $\i$ is a bottleneck for $j \in C'$, then $(1 + \epsilon)^2 v_j \ge \tilde{r}_{\i}$.
  \end{claim}
  \begin{proof}
    For $j \in C'$, we have
    \begin{align*}
      (1 + \epsilon)^2 v_j &= (1 + \epsilon)^{-2} v'_j \\
                           &= (1 + \epsilon)^{-2} \max\LR{c_{\i j}, r_{\i}} \\
                           &\ge (1 + \epsilon)^{-2} r_{\i} \ge \tilde{r}_{\i}
    \end{align*}
  \end{proof}
  
  \begin{claim} \label{cl:apdx-apx-3}
    If $\i \in F'$ is a bottleneck for $j \in C'$, then $w_{i'j} = 0$ for all $i' \in F'$, where $i' \neq \i$.
  \end{claim}
  \begin{proof}
    Assume for contradiction that $w_{i'j} > 0$, i.e., $r_{i'} \ge c_{i'j}$ for some $i' \in F', i' \neq \i$.
    
    If $r_{\i} \ge c_{\i j}$, then $v'_j = r_{\i} \le \max\LR{c_{i'j}, r_{i'}} = r_{i'}$. In this case, $c_{\i i'} \le c_{\i j} + c_{i' j} \le 2 r_{i'}$.
    
    Otherwise, if $c_{\i j} > r_{\i}$, then $v'_j = c_{\i j} \le \max\LR{c_{i'j}, r_{i'}} = r_{i'}$. Here too we have, $c_{\i i'} \le c_{\i j} + c_{i' j} \le 2r_{i'}$.
    
    In either case, $c_{\i i'} \le 2r_{i'} \le 2 \max\LR{r_{i'}, r_{\i}} \le 2 (1 + \epsilon)^2 \max\LR{\tilde{r}_{i'}, \tilde{r}_{\i}} $, which is a contradiction to Claim \ref{cl:apdx-cl9.2} since at most one of $i', \i$ can be added to $F'$.
  \end{proof}
  
  \begin{claim} \label{cl:apdx-apx-4}
    If a closed facility $\i \not\in F'$ is the bottleneck for $j \in C'$, and if an open facility $i' \in F'$ caused $\i$ to close, then $c_{i'j} \le 3(1+\epsilon)^8 \cdot v_j$.
  \end{claim}
  \begin{proof}
    \begin{align*}
      c_{i'j} &\le c_{i'\i} + c_{\i j} \tag{Triangle inequality} \\
              &\le 2(1+\epsilon)^2 \tilde{r}_{\i} + c_{\i j} \tag{$i'$ caused $i$ to close, so using Lemma \ref{lem:apdx-lem8}.} \\
              &\le 2(1+\epsilon)^2 \cdot (1+\epsilon)^2 (w_{\i j} + c_{\i j}) + c_{\i j} \tag{Using Claim \ref{cl:apdx-apx-1}.} \\
              &\le 2(1+\epsilon)^4 w_{\i j} + (2(1+\epsilon)^4 + 1) \cdot c_{\i j} \\
              &\le \max \{2(1+\epsilon)^4, 2(1+\epsilon)^4 + 1 \} \cdot v'_j \tag{Since $\i$ is the bottleneck for $j$} \\
              &\le 3(1+\epsilon)^8 \cdot v_j \tag{By definition of $v_j$}
    \end{align*}
  \end{proof}
  
  Now we state the main lemma that uses the dual variables for analyzing the cost.
  
  \begin{lemma} \label{lem:apdx-outlier-lemma}
    For any $j \in C'$, there is some $i \in F'$ such that $c_{ij} + w_{ij} \le 3(1 + \epsilon)^8 v_j$.
  \end{lemma}
  \begin{proof}
    Fix a client $j \in C'$, and let $\i$ be its bottleneck. We consider different cases. Here, $c_{ij}$ should be seen as the connection cost of $j$, and $w_{ij}$, the cost towards opening of a facility (if $w_{ij} > 0$).
    
    \begin{itemize}
    \item \textbf{Case 1:} $i \in F'$.
      
      In this case, by Claim \ref{cl:apdx-apx-3}, $w_{i'j} = 0$ for all other $i' \in F'$.
      
      If $c_{\i j} < r_{\i}$, then $w_{\i j} > 0$, and $v'_j = c_{\i j} + w_{\i j}$. Therefore, $(1 + \epsilon)^4v_j = v'_j$ pays for the connection cost as well as towards the opening cost of $\i$.
      
      Otherwise, if $c_{\i j} \ge r_{\i}$, then $w_{\i j} = 0$. Also, $w_{i'j} = 0$ for all other $i' \in F'$. Therefore, $j$ does not contribute towards opening of any facility in $F'$. Also, we have $v'_j = \max\LR{c_{\i j}, r_{\i}} = c_{\i j}$, i.e., $(1 + \epsilon)^4v_j$ pays for $j$'s connection cost.
      
    \item \textbf{Case 2:} $\i \notin F'$ and $w_{ij} = 0$ for all $i \in F'$.
      
      Let $i' \in F'$ be the facility that caused $\i$ to close. From Claim \ref{cl:apdx-apx-4}, we have $c_{i' j} \le 3(1 + \epsilon)^8 v_j$, i.e. $3(1 + \epsilon)^8 v_j$ pays for the connection cost of $j$.
      
    \item \textbf{Case 3:} $\i \notin F'$, and there is some $i' \in F'$ with $w_{i'j} > 0$. But $i'$ did not cause $\i$ to close.
      
      Since $w_{i'j} > 0$, by Claim \ref{cl:apdx-apx-1} $\tilde{r}_{i'}\ge (1 + \epsilon)^{-2} (c_{i'j} + w_{i'j})$.
      
      Let $i$ be the facility that caused $\i$ to close. Therefore, by Claim \ref{cl:apdx-apx-4}, we have $c_{i j} \le  3(1 + \epsilon)^8 v_j$.
      
      Now, since $i$ and $i'$ both belong to $F'$, by Claim \ref{cl:apdx-cl9.2}, $c_{i' i} > 2(1+\epsilon)^2 \max \LR{\tilde{r}_i, \tilde{r}_{i'}} \ge 2\tilde{r}_{i'} \ge 2(c_{i' j} + w_{i'j})$.
      
      Therefore, by triangle inequality, $c_{ij} + c_{i'j} \ge c_{i'i} > 2(c_{i' j} + w_{i'j})$ which implies $c_{ij} > c_{i'j} + 2w_{i'j}$.
      
      It follows that, $c_{i' j} + w_{i'j} \le c_{i' j} + 2w_{i'j} < c_{ij} \le 3(1 + \epsilon)^8 v_j$. Therefore, $3(1 + \epsilon)^8v_j$ pays for the connection cost of $j$ to $i'$, and its contribution towards opening of $i'$.
      
    \item \textbf{Case 4:} $\i \not\in F'$, but for $i' \in F'$ that caused $\i$ to close has $w_{i'j} > 0$.
      
      Again, since $w_{i'j} > 0$, by Claim \ref{cl:apdx-apx-1} $\tilde{r}_{i'}\ge (1 + \epsilon)^{-2} (c_{i'j} + w_{i'j})$.
      
      From Claim \ref{cl:apdx-apx-2}, we have that $(1 + \epsilon)^2 v_j \ge \tilde{r}_{\i}$. Then, since $i'$ caused $\i$ to close, we have $\tilde{r}_{\i} \ge \tilde{r}_{i'} \ge (1 + \epsilon)^{-2} (c_{i'j} + w_{i'j})$. This implies $(1 + \epsilon)^4 v_j \ge c_{i'j} + w_{i'j}$, i.e., $(1 + \epsilon)^4 v_j$ pays for the connection cost of $j$ to $i'$, and its contribution towards opening of $i'$.
    \end{itemize}  
  \end{proof}
  
  Now we are ready to prove the approximation guarantee of the algorithm.
  \begin{theorem} \label{thm:apdx-outlier-theorem}
    $\cost_e(C', F') \le 3(1 + \epsilon)^8 \cdot \cost_e(C^*_e, F^*_e) + f_{i_e}$
  \end{theorem}
  \begin{proof}
    Recall that $f_{i_e}$ denotes the cost of the most expensive facility in an optimal solution. Furthermore, notice that for any facility $i \in F' \setminus \{i^*\}$, the clients in the ball $B(i, r_i) \subseteq C'$. However, if $i^* \in F'$, some of the clients in $B(i^*, r_{i^*})$ may have been removed in the outlier determination phase, and therefore it may not get paid completely by the dual variables $v_j$. Therefore,
    
    \begin{align*}
      \cost_e(C', F') &= \sum_{j \in C'} d(j, F) + \sum_{i \in F' \setminus \{i^*\} } f_i + f_{i^*}
      \\&\le 3 (1 + \epsilon)^8 \cdot \sum_{j \in C'} v_j + f_{i^*} \tag{From Lemma \ref{lem:apdx-outlier-lemma}}
      \\&= 3(1 + \epsilon)^8 \cdot \lr{\sum_{j \in C} v_j - \q\ell} + f_{i^*} \tag{For $j \in O'$, $v_j = \q$, and $|O'| = \ell$}
      \\&\le 3(1 + \epsilon)^8 \cdot \lr{\sum_{j \in C} v_j - \q\ell} + f_{i_e} \tag{Since $i_e$ is the most expensive facility}
    \end{align*}
    Since $(v, w, \q)$ is a feasible dual solution, its cost is a lower bound on the cost of any integral optimal solution. Therefore, the theorem follows.
  \end{proof}
  
  Therefore, we can apply corollary \ref{cor:modified-cost} with $\alpha = 1 + \epsilon, \beta = 3(1 + \epsilon)^8, \gamma = 1$ to get the following approximation guarantee --
  
  \begin{theorem}
    The solution returned by Algorithm \ref{alg:outlierskmm} is a \(5 + O(\epsilon)\) approximation to the Robust Facility Location problem
  \end{theorem}
}

\subsection{The Congested Clique and MPC Algorithms} \label{subsec:robust-fl-cc-mpc}
The algorithm for Congested Clique is essentially the same as the \kmm algorithm with \(k = n\). The only technical difference is that in the \kmm, the input graph vertices are randomly partitioned across the machines. This means that even though there are \(n\) vertices and \(n\) machines, a single machine may be hosting multiple vertices. It is easy to see that the Congested Clique model, in which each machine holds exactly one vertex can simulate the \(k\)-machine algorithm with no overhead in rounds. Therefore, by substituting \(k = n\) in the running time of Theorem \ref{thm:kmm-OutlierGuarantee}, we get the following result.

\begin{theorem} \label{thm:cc-OutlierGuarantee}
  In \(O(\poly \log n)\) rounds of Congested Clique, whp, we can find a factor \(5 + O(\epsilon)\) approximate solution to the Robust \textsc{FacLoc} problem for any constant \(\epsilon > 0\).
\end{theorem}

Now we focus on the implementing the MPC algorithm. The first crucial observation is that Algorithm \ref{alg:outlierskmm} reduces the task of finding an approximate solution to the Robust \textsc{FacLoc} problem in the implicit metric setting to \(\poly \log n\) calls to a \((1 + \epsilon)\)-approximate SSSP subroutine along with some local bookkeeping. Therefore, all we need to do is efficiently implement an approximate SSSP algorithm in the MPC model.

The second fact that helps us is that Becker et al. \cite{BeckerKKLdisc17} provide a distributed implementation of their approximate SSSP algorithm in the Broadcast Congested Clique (BCC) model. The BCC model is the same as the Congested Clique model but with the added restriction that nodes can only broadcast messages in each round. Therefore we get the following simulation theorem, which follows almost immediately from Theorem 3.1 of \cite{BehezhadDHARXIV18}.

\begin{theorem}
  Let \(\mathcal{A}\) be a \(T\) round BCC algorithm that uses \(\tilde{O}(n)\) local memory at each node. One can simulate \(\mathcal{A}\) in the MPC model in \(O(T)\) rounds using \(\tilde{O}(n)\) memory per machine. 
\end{theorem}

In any \(T\) round BCC algorithm, each vertex will receive \(O(n \cdot T)\) distinct messages. The approximate SSSP algorithm of Becker et al. \cite{BeckerKKLdisc17} runs in \(O(\poly \log n/ \poly(\epsilon))\) rounds and therefore, uses \(\tilde{O}(n)\) memory per node to store all the received messages (and for local computation). Therefore, we get the following theorem.

\begin{theorem} \label{thm:mpc-OutlierGuarantee}
  In \(O(\poly \log n)\) rounds of MPC, whp, we can find a factor \(5 + O(\epsilon)\) approximate solution to the Robust \textsc{FacLoc} problem for any constant \(\epsilon > 0\).
\end{theorem}

 \section{Distributed Robust Facility Location: Explicit Metric}\label{sec:robust-fl-explicit}
For the \kmm implementation, the implicit metric algorithm from the previous section also provides a similar guarantee for the explicit metric setting and hence we do not discuss it separately in this section.

\subsection{The Congested Clique Algorithm} \label{sec:robust-fl-explicit-cc}
The work in \cite{HegemanPDC15} presents a Congested Clique algorithm that runs in
expected $O(\ll)$ rounds and computes an $O(1)$-approximation to \textsc{FacLoc}.
This is improved exponentially in \cite{HegemanPSarxiv14} which presents an
$O(1)$-approximation algorithm to \textsc{FacLoc} running in $O(\lll)$ rounds whp.
The algorithms in \cite{HegemanPDC15} and in \cite{HegemanPSarxiv14} are essentially
the same with one key difference. They both reduce the problem of solving \textsc{FacLoc}
in the Congested Clique model to the ruling set problem. Specifically, showing
that if a $t$-ruling set can be computed in $T$ rounds, then an $O(t)$-approximation
to \textsc{FacLoc} can be computed in $O(T)$ rounds. 
In \cite{HegemanPDC15} a 2-ruling set is computed in expected $O(\ll)$ rounds,
whereas in \cite{HegemanPSarxiv14} it is computed in $O(\lll)$ rounds whp.

The algorithm for computing an $O(1)$-approximation to Robust \textsc{FacLoc}
(see Section \ref{sec:robust-fl-seq}) is essentially the \textsc{FacLoc} algorithm in \cite{HegemanPDC15,HegemanPSarxiv14},
but with an outer loop that runs $O(\log n)$ times.
In each iteration of this outer loop, we modify the facility opening costs in 
a certain way and solve \textsc{FacLoc} on the resulting instance. Thus we have $O(\log n)$
instances of \textsc{FacLoc} to solve and via the reduction in \cite{HegemanPDC15,HegemanPSarxiv14},
we have $O(\log n)$ independent instances of the ruling set problem to solve.
Here we show that $O(\log n)$ independent instances of the $O(\lll)$-round 2-ruling set algorithm in
\cite{HegemanPSarxiv14} can be executed in parallel in the Congested Clique model,
still in $O(\lll)$ rounds whp.
To be precise, suppose that the input consists of $c = O(\log n)$ graphs
$G_1 = (V, E_1), G_2 = (V, E_2),  \ldots, G_c = (V, E_c)$. 

\begin{theorem} \label{theorem:parallelExecution}
2-ruling sets for all graphs $G_i$, $1 \le i \le c$, can be computed in $O(\lll)$ rounds whp.
\label{theorem:lognParallel}
\end{theorem}
\ifthenelse{\boolean{short}}{}{
  \begin{proof}
    The proof is simply an accounting of the communication that occurs in each phase of the 2-ruling set
    algorithm in \cite{HegemanPSarxiv14}.
    The accounting establishes that there is enough bandwidth in the Congested Clique model to
    allow for $c$ instances of the algorithm (one for each graph $G_i$) to run in parallel, without increasing the number
    of rounds by more than a constant-factor.
    The 2-ruling set algorithm of \cite{HegemanPSarxiv14} consists of 5 phases (in this order): (1) Lazy Degree Decomposition phase,
    (2) Speedy Degree Decomposition phase, (3) Vertex Selection phase, (4) High Degree Vertex Removal phase, (5) MIS in Low-Degree
    Graphs phase.
    Phases (1)-(4) are described in detail in \cite{HegemanPSarxiv14}, whereas Phase 5 is described in \cite{GhaffariPODC17}.
    
    In the Lazy Degree Decomposition phase and the Vertex Selection phases, each 
    vertex communicates by broadcasting a bit. To run $c$ instances of these
    two phases, a vertex can simply package the $c$ bits it needs
    to send (one for each instance) into $O(1)$ messages
    of size $O(\log n)$ each and use $O(1)$ rounds to perform the
    communication.
    
    The key part of the Speedy Degree Decomposition phase is for each vertex $v \in V$ to learn
    $B_{G_t}(v, \lceil \ll \rceil)$, the topology up to distance $\lceil \ll \rceil$ hops
    of the graph $G_t$ that is active after the  Lazy Degree Decomposition phase.
    This phase consists of $\lceil \lll \rceil$ iterations (each of which can be
    implemented in $O(1)$ rounds) and in iteration $i$, $0 \le i \le \lceil \lll \rceil -1$, vertex $v$'s knowledge expands
    from $B_{G_t}(v, 2^i)$ to $B_{G_t}(v, 2^{i+1})$.
    This is achieved by each vertex $v$ sending $B_{G_t}(v, 2^i)$ to all vertices
    in $B_{G_t}(v, 2^i)$ in iteration $i$. 
    If $G_t$ has maximum degree $\Delta$, then $B_{G_t}(v, 2^i)$ contains $O(\Delta^{2^i})$ vertices
    and $O(\Delta^{2^i + 1})$ edges.
    Thus each vertex $v$ needs to send (and receive) a total of $O(\Delta^{2^{i+1} + 2^i}) = O(\Delta^{3\ll})$ messages.
    
    If the Lazy Degree Decomposition phase is run for $t$ rounds, then $\Delta \le n^{1/2^t}$.
    Currently, the Lazy Degree Decomposition is run for $t = 1 + \lceil \lll \rceil$ iterations.
    If we run it for two additional iterations, then $\Delta \le n^{1/8\ll}$ and therefore
    the total number of messages a vertex needs to send (receive) per round in an instance of the Speedy
    Degree Decomposition phase is $O(n^{3/8})$.
    So even if we were to run $c = O(\log n)$ instances of this phase in parallel, the number of messages
    a vertex needs to send (receive) per round is $O(n)$.
    Therefore, \lra\ can be used to get all of these messages to their destination in $O(1)$ rounds and therefore
    all $c$ instances of the algorithm can complete an iteration of the Speedy Degree Decomposition phase
    in $O(1)$ rounds.
    
    The High Degree Vertex Removal phases starts with a set $S$ of vertices that are
    active after the Vertex Selection phase.
    A leader vertex (e.g., a vertex with lowest \texttt{ID}) generates a random ranking (permutation)
    of the vertices in $S$. Let the parameter $\delta = \log^2 n$. The subgraph induced by $P \subseteq S$, the set of vertices with rank
    in $[1 \ldots |S|/\delta]$, is sent to the leader.
    To run $c$ instances of this phase, we simply use $c$ distinct leaders
    (e.g., the $c$ vertices with lowest \texttt{ID}s), one for each instance of the algorithm.
    This permits the random rank generation for the $c$ instances to complete in parallel 
    in $O(1)$ rounds.
    In the proof of Theorem ??? in \cite{HegemanPSarxiv14} it is shown that the number of edges of the graph induced by $P$,
    incident on a vertex is $O(n/\log^2 n)$ whp.
    It is also shown that whp the total number of edges in this induced subgraph is $O(n)$.
    Therefore, each vertex needs to send $O(n/log^2 n)$ messages to the leader. Even with $c = O(\log n)$
    instances of the algorithm, each vertex needs to send $O(n/\log n)$ messages.
    Furthermore, since we are using $c$ distinct leaders to receive the graphs each of size $O(n)$,  we can use
    \lra\ to complete this communication in $O(1)$ rounds.

    Finally, we examine the MIS in Low-Degree Graphs phase.
    This phase consists of two parts, the first being the gathering by each vertex $v$ of
    its $O(\log\log n)$-hop neighborhood (see Lemma 2.15 in \cite{GhaffariPODC17}).
    The maximum degree of the graph on which we run this phase is $O(\log^3 n)$ and therefore
    the accounting that we did for the Speedy Degree Decomposition phase applies.
    Recall that in the analysis of the Speedy Degree Decomposition phase the maximum
    degree was bounded above by $n^{1/8\ll}$.
    Finally, in the second part of the MIS in Low-Degree Graphs phase, the graph that is
    still active is gathered and processed at a single leader vertex.
    It is shown in Lemma 2.11 in \cite{GhaffariPODC17} that this graph has $O(n)$ edges and therefore
    can gathered at the leader in $O(1)$ rounds. To run $c = O(\log n)$ instances of this
    phase, as before, we simply pick $c$ leaders.
    Thus there is still enough bandwidth from receiver's perspective.
    Also note that there is enough bandwidth from the sender's perspective because 
    the maximum degree of the graph that enters the MIS in
    Low-Degree Graphs phase is $O(\log^3 n)$.
  \end{proof}
}

\noindent
The theorem above and the discussion preceding it leads to the following theorem.
\begin{theorem}
There is an $O(1)$-approximation algorithm in the Congested Clique model for Robust \textsc{FacLoc}, running in
$O(\lll)$ rounds whp.
\end{theorem}

\subsection{The MPC Algorithm}
We now utilize the Congested Clique algorithm for Robust \textsc{FacLoc} to design an MPC model algorithm for Robust \textsc{FacLoc}, also running in $O(\lll)$ rounds whp. Since each vertex has explicit knowledge of $n$ distances, the overall memory is $O(n^2)$ words. Since the memory of each machine is \(\tilde{O}(n)\), the number of machines will be \(\tilde{O}(n)\) as well. Therefore, we can simulate the algorithm from the preceding section using Theorem 3.1 of \cite{BehezhadDHARXIV18} in the MPC model. We summarize our result in the following theorem.

\begin{theorem}
There is an $O(1)$-approximation algorithm for Robust \textsc{FacLoc} that can be implemented
in the MPC model with $\tilde{O}(n)$ words per machine in $O(\lll)$ rounds whp.
\end{theorem}

 \ifthenelse{\boolean{short}}{}{
  \section{Facility Location with Penalties: Implicit Metric}

\subsection{The \texorpdfstring{$k$}{k}-Machine Algorithm}
In this section we describe how to implement Algorithm \ref{alg:penaltiesSeq} in the \kmm. Since the radius computation phase for \textsc{FacLoc} with Penalties is different from the one for Robust \textsc{FacLoc} (Algorithm \ref{alg:RC}), we first show how to modify Algorithm \ref{alg:RC} in order to compute approximate radii for the Penalty version.

\subsubsection{Radius Computation}
\label{sec:pen-imp-radcomp}
In \textsc{FacLoc} with Penalties, the definition of radii differs from that in Robust \textsc{FacLoc} (or the standard Facility Location algorithm) due to the penalties of the clients. In particular, for a vertex $v$, the radius is defined as $r_v$ satisfying the following equation: $f_v = \sum_{u \in V} \max\{\min\{r_v - d(u, v), p_u - d(u, v) \}, 0\}$. Throughout this section, we assume that such an $r_v$ exists -- otherwise, it can be shown that excluding $v$ as a candidate facility does not affect the cost of any solution. We now show how to appropriately modify the neighborhood computation and the radius computation subroutines. 

The key idea is to divide vertices into $O(\log n)$ classes, such that the penalties of the vertices belonging to a particular class are within $1+O(\epsilon)$ factor of each other. Then, for any vertex $v$, and for each penalty class, we estimate the number of vertices from that penalty class, in $(1+\epsilon)^i$-neighborhood of $v$. Once we have these estimates for each range of neighborhoods, they can be used for computation of approximate computation of radii. We formalize this in the following.

First, we assume that we have normalized $f_i, c_{ij}, p_j$, such that any positive quantity is at least $1$. Note that we can normalize any given input in this manner in $O(1)$ rounds of the \kmm. Let $P_0 \coloneqq \{j \in V \mid p_j = 0 \}$, and for any integer $t \ge 1$, let $P_b \coloneqq \{j \in V \mid (1+\epsilon)^b \le p_j < (1+\epsilon)^{b+1}\}$. By assumption, the penalties are polynomially bounded in $n$, and hence the total number of penalty classes is $O(\log n)$.

Let \textsc{NbdSizeEstimates}($G, \epsilon, b$) be a modified version (of the original algorithm, Algorithm 3 in \cite{BandyapadhyayArxiv2017}, see Algorithm \ref{alg:ModifiedCohenEstimates}) that takes an additional parameter $b \ge 0$, wherein only the vertices in $P_b$ participate. That is, random ranks (as in the original version) are chosen only for the vertices in $P_b$. However, the neighborhood size estimates are computed for \emph{all} vertices. The details are straightforward, and are therefore omitted. 

\RestyleAlgo{boxruled}
\begin{algorithm2e}\caption{\textsc{NbdSizeEstimates}\((G, \epsilon, b)\)\label{alg:ModifiedCohenEstimates}}
	$\epsilon' := \epsilon/(\epsilon + 4)$; $t = \lceil 2 \log_{1+\epsilon'} n \rceil$; $\ell := \lceil c \log n/(\epsilon')^2 \rceil$\\
	\For{\(j := 1, \ldots, \ell\)} {
		\textbf{Local Computation.} Each machine $m_j$ picks a rank $\r(v)$, for each vertex $v \in H(m_j) \cap P_b$,
		chosen uniformly at random from $[0, 1]$. Machine $m_j$ then rounds $\r(v)$ down to the closest 
		$(1+\epsilon')^i/n^2$ for integer $i \ge 0$\label{alg2.1:ChooseRanks} \\
		\For{$i := 0, 1, \ldots, t-1$}{
			$T_i := \{v \in P_b \mid \r(v) = (1+\epsilon')^i/n^2\}$\\
			Compute a \((1+\epsilon)\)-approximate solution to MSSP using \(T_i\) as the set of sources \label{alg2.1:MSSP}; let $\tilde{d}(v, T_i)$ denote the computed approximate distances\\
			\textbf{Local Computation.} Machine $m_j$ stores $\tilde{d}(v, T_i)$ for each \(v \in H(m_j)\)\\
		}
	}
\end{algorithm2e}

For a vertex $v \in V$, and for parameters $b, r \ge 0$, let $B(v, r, b) \coloneqq B(v, r) \cap P_b$. Then, let $Q(v, r, b)$ denote the query ``What is the size of $B(v, r) \cap P_b$?''. The details of how to answer this query from the output of \textsc{NbdSizeEstimates} are same as in the original version.

\begin{lemma} \label{lem:nbd-with-penalties}
	For any vertex $v \in V$, for any $r, b \ge 0$, and for any $\epsilon > 0$, the modified \textsc{NbdSizeEstimates} algorithm satisfies the following properties.
	\begin{itemize}
		\item For the query $Q(v, r/(1+\epsilon), b)$, the algorithm returns an output that is at most $|B(v, r) \cap P_b| \cdot (1+\epsilon)$.
		\item For the query $Q(v, r(1+\epsilon), b)$, the algorithm returns an output that is at most $|B(v, r) \cap P_b|/(1+\epsilon)$.
	\end{itemize}
\end{lemma}

We define the following quantities with respect to any vertex $v \in V$. Let $\alpha(v, r, b) \coloneqq \sum_{u \in B(v, r, b)} \max\{\min\{r - d(u, v), p_u - d(u, v)\}, 0 \}$, and let $\alpha(v, r) \coloneqq \sum_{u \in B(v, r)} \max\{\min\{r - d(u, v), p_u - d(u, v)\}, 0 \}$. It is easy to see that $\alpha(v, r) = \sum_{b \ge 0} \alpha(v, r, b)$. Finally, let $q_{i, b}(v) \coloneqq |B(v, (1+\epsilon)^i, b)|$.

\begin{lemma} \label{lem:alpha-lower-bound}
	If $t > b$, then $\alpha(v, (1+\epsilon)^t, b) \ge \sum_{i = 0}^{t-1} q_{i, b} \lr{(1+\epsilon)^{i+1} - (1+\epsilon)^i}$. Otherwise, $\alpha(v, (1+\epsilon)^t, b) \ge \sum_{i = 0}^{b-1}\ q_{i, b}(v)\lr{(1+\epsilon)^b - (1+\epsilon)^{i+1}} + \sum_{i = b}^{t-1} q_{i, b} \lr{(1+\epsilon)^{i+1} - (1+\epsilon)^i}$.
\end{lemma}
\begin{proof}
	First, let $t > b$. And consider,
        \begingroup
        \allowdisplaybreaks
	\begin{align*}
		\alpha(v, (1+\epsilon)^t, b) &= \sum_{i = 0}^{t -1} \alpha(v, (1+\epsilon)^{i+1}, b) - \alpha(v, (1+\epsilon)^i, b)
		\\&\ge \sum_{i = 0}^{t - 1} \sum_{u \in B(v, (1+\epsilon)^i, b)} \bigg( \min\{(1+\epsilon)^{i+1} - d(u, v), (1+\epsilon)^b - d(u, v)\} \\
		&\qquad - \min\{(1+\epsilon)^{i} - d(u, v), (1+\epsilon)^{b+1} - d(u, v)\}\bigg)
		\\&= \sum_{i = 0}^{t-1} \sum_{u \in B(v, (1+\epsilon)^i, b)} (1+\epsilon)^{i+1} - d(u, v) - ((1+\epsilon)^i - d(u, v))
		\\&= \sum_{i = 0}^{t-1} \sum_{u \in B(v, (1+\epsilon)^i)} (1+\epsilon)^{i+1} - (1+\epsilon
		)^i
		\\&= \sum_{i = 0}^{t-1}q_{i, b} \lr{(1+\epsilon)^{i+1} - (1+\epsilon)^i}.
	\end{align*}
	Now, if $t \le b$, then
	\begin{align*}
	\alpha(v, (1+\epsilon)^t, b) &= \alpha(v, (1+\epsilon)^b, b) + \sum_{i = b}^{t -1} \alpha(v, (1+\epsilon)^{i+1}, b) - \alpha(v, (1+\epsilon)^i, b)
	\\&\ge \alpha(v, (1+\epsilon)^b, b) + \sum_{i = b}^{t-1} q_{i, b} \lr{(1+\epsilon)^{i+1} - (1+\epsilon)^i}	
	\end{align*}
        \endgroup
	Where final step uses similar arguments from the previous case. Now, we consider,
	\begin{align*}
		\alpha(v, (1+\epsilon)^b, b) &= \sum_{u \in B(v, (1+\epsilon)^b, b)} (1+\epsilon)^b - d(u, v) 
		\\&= \sum_{i = 0}^{b-1}\   \sum_{u \in B(v, (1+\epsilon)^{i+1}, b) \setminus  B(v, (1+\epsilon)^i, b)} (1+\epsilon)^b - d(u, v)
		\\&\ge \sum_{i = 0}^{b-1}\ q_{i, b}(v)\lr{(1+\epsilon)^b - (1+\epsilon)^{i+1}}
	\end{align*}
	
\end{proof}

For a vertex $v$, and for any $t, b \ge 0$, let $\lambda(v, t, b)$ denote the appropriate lower bound on $\alpha(v, (1+\epsilon)^t, b)$, given by Lemma \ref{lem:alpha-lower-bound} (based on two different cases). Let $\tilde{\lambda}(v, t, b)$ be the quantity obtained by replacing $q_{i, b}(v)$ by the approximate neighborhood estimate $\tilde{q}_{i, b}(v)$ in the lower bound $\lambda(v, t, b)$. We now state the Radius Computation algorithm.

\RestyleAlgo{boxruled}
\begin{algorithm2e}\caption{\textsc{RadiusComputation} Algorithm (Version 2)\label{alg:RC2}}
	\textbf{Neighborhood-Size Computation.} Call the \textsc{NbdSizeEstimates} algorithm (Algorithm \ref{alg:CohenEstimates}) to obtain \textit{approximate} neighborhood-size estimates $\tilde{q}_{i, b}(v)$
	for all integers $i \ge 0, b \ge 0$ and for all vertices $v$.\\
	\textbf{Local Computation.} Each machine $m_j$ computes $\tilde{r}_v$ locally, for all vertices $v \in H(m_j)$ using the formula $\tilde{r}_v := (1+\epsilon)^{t-1}$ where $t \ge 1$ is the smallest integer for which $\sum_{b \ge 0} \tilde{\lambda}(v, t, b) > f_v$. If there is no such integer, define $\tilde{r}_v = \infty$.
\end{algorithm2e}

We have the following bounds on the approximate radius computed by the algorithm.

\begin{lemma} \label{lem:penalty-approx-r}
	For every $v \in V$, $\frac{r_v}{(1+\epsilon)^2} \le \tilde{r}_v \le (1+\epsilon)^2 \cdot r_v$
\end{lemma}
\begin{proof}
	 By Lemma \ref{lem:nbd-with-penalties}, we have the following for $i \ge 1$:
	$$\frac{1}{(1+\epsilon)} \cdot q_{i-1, b}(v) \le \tilde{q}_{i, b}(v) \le (1+\epsilon) \cdot q_{i+1, b}(v).$$

	Now, let $t$ be the smallest integer for which $\sum_{b \ge 0}\tilde{\lambda}(v, t, b) > f_v$. Now, by Lemma \ref{lem:alpha-lower-bound}, we have that $\lambda(v, t, b) \ge \alpha(v, (1+\epsilon)^t, b)$. Now, using arguments very similar to those in the proof of Lemma 8 of \cite{BandyapadhyayArxiv2017}, one can show the following inequality.
	
	$$\alpha(v, (1+\epsilon)^{t-2}) \le \sum_{b \ge 0}\tilde{\lambda}(v, t, b) \le \alpha(v, (1+\epsilon)^{t+1}).$$
	
	Recall that $\alpha(v, r) = \sum_{b \ge 0} \alpha(v, r, b)$.
	
	Therefore, there must exist a value $r_v \in [(1+\epsilon)^{t-3}, (1+\epsilon)^{t+1}]$, such that $\alpha(v, r_v) = f_v$. The Lemma follows, because $\tilde{r}_v = (1+\epsilon)^{t-1}$.
\end{proof}

\begin{algorithm2e}\caption{\textsc{PenaltyFacLoc}\((G, F, C, p)\)}\label{alg:penaltieskmm}
  \tcc{Radius Computation Phase:}
  Compute $r_i$ for each $i \in F$ satisfying $f_i = \sum_{j \in C} \max\LR{\min\LR{r_i - c_{ij}, p_j - c_{ij}}, 0}.$ \\
  \tcc{Greedy Phase:}
  $C' \gets \emptyset$, \quad $F' \gets \emptyset$, \quad $O' \gets \emptyset$. \\
  \For{\(i = 0, 1, 2, \dots\)}{
    Let \(W\) be the set of vertices \(w \in F\) across all machines with \(\tilde{r}_w = \tilde{r} = (1 + \epsilon)^i\) \\
    Using Lemma \ref{lem:MSSP}, remove all vertices from \(W\) within distance \(2(1 + \epsilon)^2\cdot \tilde{r}\) from \(F'\) \label{alg5:remove-close-facs}\\
    \(I \leftarrow \textsc{ApproximateMIS}(G, W, 2(1 + \epsilon)^3 \cdot \tilde{r}, \epsilon)\) \\
    \(F' \leftarrow F' \cup I\)
  }
  \tcc{Outlier Determination Phase:}
  Using Lemma \ref{lem:MSSP}, add to \(C'\) all clients \(j\) having distance to \(F'\) less than \(p_j\) and add the rest to \(O'\) \label{alg5:remove-clients}\\
  \Return $(C', F')$ as the solution.
\end{algorithm2e}

Our \kmm implementation of the \textsc{FacLoc} with Penalties algorithm is summarized in Algorithm \ref{alg:penaltieskmm}. The correctness proof is similar to that of Algorithm \ref{alg:penaltiesSeq} but is complicated by the fact that we compute \((1 + \epsilon)\)-approximate distances instead of exact distances. This analysis appears in the next section, and as a result we get the following theorem the proof of which is similar to Theorem \ref{thm:kmm-OutlierGuarantee}.

\begin{theorem} \label{thm:kmm-PenaltyGuarantee}
	In \(\tilde{O}(\text{poly}(1/\epsilon) \cdot n/k)\) rounds, whp, Algorithm \ref{alg:penaltieskmm} finds a factor \(5 + O(\epsilon)\) approximate solution \((C', F')\) to the \textsc{FacLoc} with Penalties problem.
\end{theorem}

\subsubsection{Analysis of the Algorithm}

We state the standard primal and dual linear programming relaxations for facility location with penalties in Figure \ref{fig:penalty-primal-dual}. For $j \in C$ and $i \in F$, define $w_{ij} \coloneqq \max\LR{\min \LR{r_i - c_{ij}, p_j - c_{ij}}, 0}$ and for $j \in C$, let $v'_j \coloneqq \min_{i \in F} c_{ij} + w_{ij}$. Note that $v'_j = \min_{i \in F}\max\LR{c_{ij}, \min\LR{r_i, p_j}}$. If $\i$ is a facility realizing the minimum $v'_j = c_{\i j} + w_{\i j}$, then we say that $\i$ is the bottleneck of $j$.

To make the analysis easier, we consider a more costly solution $(\tilde{C}', F')$ where the set of clients $\tilde{C}'$ is constructed using the following modified outlier determination phase: for each client $j$, if $c_{\ij} \le p_j$ then $j \in \tilde{C}'$ and otherwise $j \in \tilde{O}'$ where $\i$ is the bottleneck of $j$.

It is easy to see by an exchange argument that the $\cost_e(C', F') \le \cost_e(\tilde{C}', F')$, the outliers determined in the algorithm are at least as far from \(F'\) as ones in the modified outlier determination phase. Henceforth, we analyze the cost of the solution $(\tilde{C}', F')$ by comparing it to the cost of a feasible dual \LP solution and in order to alleviate excessive notation, we will henceforth refer to the solution $(\tilde{C}', F')$ as $(C', F')$ and $\tilde{O}'$ as $O'$.

Because of the way we choose the outliers in the solution we consider for the analysis \((C', F')\) we have the following property (where $\i$ is the bottleneck of $j$) --
$$v_j = \begin{cases}
  \max\LR{c_{\ij}, r_\i} &\text{if } j \in C' \\
  p_j  &\text{if } j \in C \setminus C'
\end{cases}$$

Throughout this section, we condition on the event that the outcome of all the randomized algorithms is as expected (i.e. the ``bad'' events do not happen). Note that this happens with w.h.p. We first need the following facts along the lines of \cite{Thorup2001}. We skip the proofs as they are identical to the corresponding facts in the previous section.

\begin{lemma}[Modified From Lemma 8 Of \cite{Thorup2001}] \label{lem:apdx-p-lem8}
  There exists a total ordering $\prec$ on the facilities in $F$ such that $u \prec v \implies \tilde{r}_u \le \tilde{r}_v$, and $v$ is added to $F'$ if and only if there is no previous $u \prec v$ in $F'$ such that $c_{u v} \le 2(1+\epsilon)^2 \tilde{r}_v$.
\end{lemma}

\begin{claim}[Modified From Claim 9.2 Of \cite{Thorup2001}] \label{cl:apdx-p-cl9.2}
  For any two distinct vertices $u, v \in F'$, we have that $c_{u v} > 2(1+\epsilon)^2 \cdot \max\{\tilde{r}_u, \tilde{r}_v\}$.
\end{claim}

We simultaneously prove feasibility of the dual solution we constructed, and show how it can be used to pay for the integral solution. We consider different cases regarding a fixed client $j \in C$ with bottleneck facility $\i$. We first prove a few straightforward claims.

\begin{claim} \label{cl:apdx-p-4}
  For any $i \in F$ and $j \in C$, $\tilde{r}_i \le (1 + \epsilon)^2 (c_{ij} + w_{ij})$. Furthermore if for some \(i \in F\) and \(j \in C\), \(w_{ij} > 0\), then \(\tilde{r}_i \ge (1 + \epsilon)^{-2} (c_{ij} + w_{ij})\)
\end{claim}
\begin{proof}
  We have $w_{ij} = \max\{0, \min\LR{r_i, p_j} - c_{ij}\} \ge \min\LR{r_i, p_j} - c_{ij}$ and therefore $\tilde{r}_i \le (1 + \epsilon)^2 r_i \le (1 + \epsilon)^2 (c_{ij} + w_{ij})$.

  If for some \(i \in F\) and \(j \in C\), \(w_{ij} > 0\), then \(w_{ij} = r_i - c_{ij}\) which means that \(\tilde{r}_i \ge (1 + \epsilon)^{-2} r_i = (1 + \epsilon)^{-2} (c_{ij} + w_{ij})\)
\end{proof}

\begin{claim} \label{cl:apdx-p-1}
  If $\i \in F'$ is the bottleneck for $j \in C$, then $w_{i'j} = 0$ for all $i' \in F'$, where $i' \neq \i$.
\end{claim}
\begin{proof}
  Suppose there exists a facility $i' \in F'$ with $w_{i'j} > 0$. That is, $\min\LR{r_{i'} - c_{i'j}, p_j - c_{i'j}} > 0$, which further implies that $c_{i'j} < \min\LR{p_j, r_{i'}} \le r_{i'} \le (1 + \epsilon)^2 \tilde{r}_{i'}$.

  If $c_{\ij} \ge r_{\i}$, then $v'_j = c_{\ij} + \max\LR{0, \min\LR{r_{\i} - c_{\ij}, p_j - c_{\ij} }} = c_{\ij}$. However, since $w_{i'j} > 0$, $c_{\ij} = v'_j \le c_{i'j} + w_{i'j} = \min\LR{p_j, r_{i'}} \le r_{i'} \le (1 + \epsilon)^2 \tilde{r}_{i'}$

  Otherwise, $c_{\ij} <  r_{\i} \le (1 + \epsilon)^2 \tilde{r}_{\i}$.

  Therefore in either case, $c_{\i i'} \le c_{\ij} + c_{i'j} \le 2 (1 + \epsilon)^2 \max\LR{\tilde{r}_{i'}, \tilde{r}_{\i}}$, which is a contradiction to Claim \ref{cl:apdx-p-cl9.2}.
\end{proof}

\begin{claim} \label{cl:apdx-p-2}
  If $\i \not\in F'$ is the bottleneck of $j \in C$ and $\max\LR{r_\i, c_{\ij}} \le p_j$, and $i' \in F'$ caused $\i$ to close, then $c_{i' j} \le 3(1+\epsilon)^4 v_j$.
\end{claim}
\begin{proof}
  Note that since we assume $\max\LR{r_\i, c_{\ij}} \le p_j$, we have $j \in C'$, $v_j = v'_j$, $c_{\ij} \ge p_j$, and $r_\i \le p_j$.
  Since $i' \in F'$ caused $\i$ to close, $\tilde{r}_{i'} \le \tilde{r}_{\i}$. Furthermore, $c_{i' \i} \le 2 (1+\epsilon)^2 \tilde{r}_{\i}$. Therefore, $c_{i' j} \le c_{\i i'} + c_{\i j} \le 2 (1+\epsilon)^2 \tilde{r}_{\i} + c_{\ij} \le 2 (1+\epsilon)^4 r_{\i} + c_{\ij} $.

  If $c_{\ij} \ge r_{\i}$, then $w_{\ij} = 0$, and $c_{\ij} = v_j$. This means $2(1+\epsilon)^4r_{\i} + c_{\ij} \le 3(1+\epsilon)^4c_{\ij} = 3(1+\epsilon)^4v_{j}$
  Otherwise, $c_{\ij} < r_{\i}$. Here, $v_j = \min\LR{r_{\i}, p_j}$ Then $2(1+\epsilon)^4r_{\i} + c_{\ij} < 3(1+\epsilon)^4r_{\i} = 3(1+\epsilon)^4 \min\LR{r_{\i}, p_j} \le 3(1+\epsilon)^4v_{j}$

  In either case, $c_{i' j} \le 3(1+\epsilon)^4v_{j}$.
\end{proof}

\begin{claim} \label{cl:apdx-p-3}
  If $\i$ is the bottleneck of $j$, with $\max\LR{r_\i, c_{\ij}} \le p_j$, then $(1 + \epsilon)^2 v'_j \ge \tilde{r}_\i$.
\end{claim}
\begin{proof}
  Again, since we assume $\max\LR{r_\i, c_{\ij}} \le p_j$, we have $j \in C'$, $v_j = v'_j$, $c_{\ij} \ge p_j$, and $r_\i \le p_j$.

  Recall that $v_j = v'_j = \max\LR{c_{\ij}, \min\LR{r_\i, p_j}} = \max\LR{c_{\ij}, r_{\i}} \ge r_{\i} \ge (1 + \epsilon)^{-2} \tilde{r}_{\i}$ and the claim follows.
\end{proof}

We are now ready to prove the feasibility and approximation guarantee

\begin{lemma}
  The solution $(v, w)$ is a feasible solution to the dual LP relaxation \ref{fig:penalty-primal-dual}.
\end{lemma}
\begin{proof}
  First note that constraints \ref{constr:penlpd-beta}, \ref{constr:penlpd-penalty}, and \ref{constr:penlpd-integral} are satisfied by construction for all $i \in F$ and $j \in C$ and so is constraint \ref{constr:penlpd-alpha} for all $j \in C'$.

  All that is left to show is that constraint \ref{constr:penlpd-alpha} is satisfied for all $j \in O'$. Since $j \in O'$, $\max\LR{r_\i, c_{\ij}} > p_j$.

  We have, $v_j \coloneqq p_j < \max\LR{r_\i, c_{\ij}} =  \max\LR{c_{\ij}, \min\LR{r_{\i}, p_j}} = v'_j \le c_{ij} + w_{ij}$ for any $i \in F$.
\end{proof}

\begin{lemma}
  For any $j \in C'$, there is some $i \in F'$ such that $c_{ij} + w_{ij} \le 3 (1+\epsilon)^4 v_j$.
\end{lemma}
\begin{proof}
  In all the cases, we assume that $\max\LR{r_{\i}, c_{\ij}} \le p_j$ and therefore $j \in C'$. This also implies $v_j = v'_j = \max \LR{c_{\ij}, \min\LR{p_j, r_{\i}}} = \max \LR{c_{\ij}, r_{\i}}$. Therefore, we can just disregard the penalties in the analysis.
  \begin{itemize}
      
    \item \textbf{Case 1.} $\i \in F'$
      
      Connect $j$ to $\i$. From Claim \ref{cl:apdx-p-1}, we know that $w_{i'j} = 0$ for all other $i' \in F'$.
      
      \begin{enumerate}
      \item If $c_{\ij} < r_{\i}$, then $v_j = c_{\ij} + w_{\ij}$. In this case, $v_j$ pays for connecting $j$ to $\i$ and also for $j$'s contribution to opening cost of $\i$ which is exactly $w_{\ij}$.
      
      \item Otherwise $c_{\ij} \ge r_{\i}$, then $w_{\ij} = 0$, which is $j$'s contribution towards $\i$. We have $v_j = c_{\ij}$ and therefore $v_j$ pays for connecting $j$ to $\i$.
      \end{enumerate}
      
    \item \textbf{Case 2.} $\i \notin F'$ and $w_{ij} = 0$ for all $i \in F'$.
      
      Let $i'$ be the facility that caused $\i$ to close. Connect $j$ to $i'$. From Claim \ref{cl:apdx-p-2}, we have $c_{i' j} \le 3(1+\epsilon)^4 v_{j}$. Therefore, $3(1+\epsilon)^4 v_j$ pays for the connection to $i'$.
      
    \item \textbf{Case 3.} $\i \notin F'$, there is some $i' \in F'$ with $w_{i'j} > 0$, but $i'$ did not cause $\i$ to close.
      
      We connect $j$ to $i'$. By assumption $w_{i'j} = r_{i'} - c_{i'j} > 0$. Furthermore, let $i$ be the facility that caused $\i$ to close. By Claim \ref{cl:apdx-p-2} we have $c_{ij} \le 3(1+\epsilon)^4 v_j$.
      
      We have $c_{i'j} + w_{i'j} = r_{i'}$. Also, $c_{i i'} > 2(1+\epsilon)^2 \tilde{r}_{i'}$, by Claim \ref{cl:apdx-p-cl9.2} since $i', i$ both were added to $F'$.
      
      Now, $2(c_{i' j} + w_{i'j}) = 2r_{i'} \le 2(1+\epsilon)^2 \tilde{r}_{i'} < c_{i i'} \le c_{i' j} + c_{i j}$.

      Subtracting $c_{i'j}$ from both sides, we get $c_{i' j} + 2w_{i'j} \le c_{i j} \le 3(1+\epsilon)^4 v_j$. Therefore, $3(1+\epsilon)^4 v_j$ pays for the connection cost of $j $ to $i'$ and also for (twice) $j$'s contribution towards opening $i'$.
      
    \item \textbf{Case 4.} $\i \notin F'$ and $i' \in F'$ with $w_{i'j} > 0$ caused $\i$ to close.
      
      We connect $j$ to $i'$. From Claim \ref{cl:apdx-p-3}, we have that $(1+\epsilon)^2 v_j \ge r_\i$.
      
      Since $i'$ caused $\i$ to close, $\tilde{r}_\i \ge \tilde{r}_{i'} \ge (1+\epsilon)^{-2} r_{i'} \ge (1+\epsilon)^{-2} c_{i'j} + w_{i'j}$.

      Therefore, $c_{i'j} + w_{i'j} \le (1+\epsilon)^2\tilde{r}_{\i} \le (1+\epsilon)^4 r_{\i} \le (1+\epsilon)^4v_j$. That is, $(1+\epsilon)^4v_j$ pays for the connection cost of $j$ to $i'$, as well as its contribution towards opening of $i'$.
    \end{itemize}
\end{proof}

Thus, $(v, w)$ is a feasible dual solution. We use the above analysis to conclude with the following theorem.

\begin{theorem}
  $$\cost(C', F') \le 3(1+\epsilon)^4 \cdot \cost(C^*, F^*).$$
\end{theorem}
\begin{proof}
  We show $\cost(C', F') \le 3(1+\epsilon)^4 \lr{\sum_{j \in C} v_j}$, which is sufficient since $(v, w)$ is a feasible dual solution, and cost of any feasible dual solution is a lower bound on the cost of an integral optimal solution.

  As we have argued previously, for any $j \in C \setminus C'$, we have $p_j = v_j$, and that for any $j \in C'$, we have $d(j, F') + s(j)$, where $s(j) \ge 0$ is the contribution of $j$ towards opening a single facility in $F'$. We have also argued that any $j \in C'$ contributes $s(j)$ for at most one open facility from $F'$. It follows that,

  \begin{align*}
    \cost(C', F') &= \sum_{i \in F'} f_i + \sum_{j \in C'} d(j, F') + \sum_{j \in C \setminus C'} p_j
    \\&= \lr{\sum_{j \in C'} s(j) + d(j, F')} + \sum_{j \in C \setminus C'} p_j
    \\&\le 3(1+\epsilon)^4 \sum_{j \in C'} v_j + \sum_{j \in C \setminus C'} v_j
    \\&\le 3(1+\epsilon)^4 \lr{\sum_{j \in C} v_j}
  \end{align*}
\end{proof}

\subsection{The Congested Clique and MPC Algorithms}
As argued in Section \ref{subsec:robust-fl-cc-mpc}, the Congested Clique model is essentially the same as the \kmm, where $k= n$. Plugging this into Theorem \ref{thm:kmm-PenaltyGuarantee}, we get the following theorem.

\begin{theorem} \label{thm:cc-PenaltyGuarantee}
	In \(O(\poly \log n)\) rounds of Congested Clique, whp, we can find a factor \(3+ O(\epsilon)\) approximate solution to the \textsc{FacLoc} with Penalties problem for any constant \(\epsilon > 0\).
\end{theorem}

In order to compute the radii of the facilities, the machines need to know the penalties of all the clients which can be done in $O(1)$ rounds of MPC since each machine needs to receive \(n\) words corresponding to the penalties of each client. The rest of the MPC algorithm implementation is similar to the corresponding implementation for the Robust \textsc{FacLoc} problem (Section \ref{subsec:robust-fl-cc-mpc}) so we don't repeat them again. The only difference is that we are trying to implement Algorithm \ref{alg:penaltieskmm} instead. We summarize this result in the following theorem.

\begin{theorem} 
	In \(O(\poly \log n)\) rounds of MPC, whp, we can find a factor \(3 + O(\epsilon)\) approximate solution to the \textsc{FacLoc} with Penalties problem for any constant \(\epsilon > 0\).
\end{theorem}
   \section{Facility Location with Penalties: Explicit Metric}\label{sec:penalties-explicit}
For the \kmm implementation, the implicit metric algorithm from the previous section also provides a similar guarantee for the explicit metric setting and hence we do not discuss it separately in this section.

\subsection{The Congested Clique Algorithm}

In this section, we briefly sketch how to implement the Facility Location with Penalties algorithm from Section \ref{sec:seq-fl-wp} in $O(\lll )$ rounds of the Congested Clique, in the explicit metric setting. Recall that in this setting, each vertex (i.e.) $v \in V$ knows $d(u, v)$ for all vertices $u \in V$. 

At the beginning of the algorithm, each client $j \in C$ broadcasts its penalty $p_j$ -- this takes $O(1)$ rounds. Once each facility $i \in F$ knows penalty $p_j$ of each client $j \in C$, it can locally compute $r_i$ satisfying $f_i = \sum_{j \in C} \max\{ \min \{r_i - c_{ij}, p_j - c_{ij} \}, 0 \}$. As argued in Section \ref{sec:seq-fl-wp}, it is without loss of generality to assume that for any facility $i \in F$ an $r_i$ satisfying this equation exists. This completes the radius computation phase.

The details of the greedy phase are similar to that from the Facility Location algorithms of \cite{HegemanPDC15, HegemanPSarxiv14}, (see also Section \ref{sec:robust-fl-explicit-cc}), where the computation of this phase is reduced to $3$-ruling set computation. As argued in Section \ref{sec:robust-fl-explicit-cc}, this can be done in $O(\lll)$ rounds. In fact, for the Facility Location with Penalties problem, this is simpler since each vertex participates in at most one ruling set computation, as opposed to $O(\log n)$ different ruling sets as in the Robust Facility Location. It can be shown in a similar way that this results in an $O(1)$ approximation. We summarize our result in the following theorem.

\begin{theorem} \label{thm:flwp-explicit-cc}
	There is an $O(1)$-approximation algorithm in the Congested Clique model for \textsc{FacLoc} with Penalties, running in $O(\lll)$ rounds whp.
\end{theorem}

\subsection{The MPC Algorithm}
Since each vertex has explicit knowledge of $n$ distances, the overall memory is $O(n^2)$. Therefore, we can simulate the Congested Clique algorithm from the preceding section using Theorem 3.1 of \cite{BehezhadDHARXIV18} in the MPC model. We summarize our result in the following theorem.
\begin{theorem} 
	There is an $O(1)$-approximation algorithm for \textsc{FacLoc} with Penalties that can be implemented in the MPC model with $O(n)$ words per machine in $O(\lll)$ rounds whp.
\end{theorem} }
\section{Conclusion and Open Questions}

This paper presents fast $O(1)$-factor distributed algorithms for
Facility Location problems that are robust to outliers. These algorithms
run in the Congested Clique model and two models of large-scale computation,
namely, the MPC model and the $k$-machine model. As far as we know these are the
the first such algorithms for these important clustering problems.

Fundamental questions regarding the optimality of our results remain open.
In the explicit metric setting, we present algorithms in the Congested
Clique model and the MPC model that run in $O(\lll)$ rounds.
While these may seem extremely fast, it is not clear that they are
optimal. Via the results of
Drucker et al.~\cite{drucker2012task}, it seems like showing a non-trivial
lower bound
in the Congested Clique model is out of the question for now.
So a tangible question one can ask is whether we can further improve
the running time of the 2-ruling set algorithm in the Congested Clique
model, possibly solving it in $O(\log^* n)$ or even $O(1)$ rounds.
This would immediately imply a corresponding improvement in the
running time of our Congested Clique and MPC model algorithms in
the explicit metric setting.

All the $k$-machine algorithms we present in the paper run in
$\tilde{O}(n/k)$ rounds.  It is unclear if this is optimal.  In
previous work \cite{BandyapadhyayArxiv2017}, we showed a lower bound
of $\tilde{\Omega}(n/k)$ in the implicit metric setting, assuming that
in the output to facility location problems every open facility needed
to know all clients that connect to it. The lower bound heavily relies
on the implicit metric and the output requirement
assumptions. However, even if we relax both of these assumptions,
i.e., we work in the explicit metric setting and only ask that every
client know the facility that will serve it, we still seem to be
unable to get over the $\tilde{O}(n/k)$ barrier. 
\bibliography{references}
\ifthenelse{\boolean{short}}{
}
{}

\end{document}